\documentclass[journal]{IEEEtran}
\ifCLASSINFOpdf
\else
\fi
\usepackage{amsmath}
\usepackage{amsthm} 
\usepackage{amsfonts}
\usepackage{url}
\usepackage{hyperref}
\hypersetup{
colorlinks=true,
linkcolor=blue,
filecolor=green,      
urlcolor=cyan,
}
\usepackage{tcolorbox}
\usepackage{algorithm}
\usepackage{algcompatible}
\newtheorem{mylem}{Lemma}
\newtheorem{mydef}{Definition}
\newtheorem{myexam}{Example}
\newtheorem{mythm}{Theorem}
\usepackage{multirow}
\begin{document}
\title{Earned Benefit Maximization in Social Networks Under Budget Constraint}

\author{Suman Banerjee,
        Mamata~Jenamani,~\IEEEmembership{Member,~IEEE,}
        Dilip Kumar Pratihar,~\IEEEmembership{Senior~Member,~IEEE}% <-this % stops a space
\IEEEcompsocitemizethanks{\IEEEcompsocthanksitem Dr. Suman Banerjee is with the Department of Computer Science and Engineering, Indian Institute of Technology,Gandhinagar-382355, India. E-mail: suman.b@iitgn.ac.in
\IEEEcompsocthanksitem Prof. Mamata Jenamani is with the Department of Industrial and Systems Engineering, Indian Institute of Technology,Kharagpur-721302, India. Email: mj@iem.iitkgp.ac.in
\IEEEcompsocthanksitem Prof. Dilip Kumar Pratihar is with the Department
of Mechanical Engineering, Indian Institute of Technology,Kharagpur-721302, India.}% <-this % stops a space
\thanks{Manuscript received April 19, 2005; revised August 26, 2015.}}


\markboth{Journal of \LaTeX\ Class Files,~Vol.~14, No.~8, August~2015}%
{Shell \MakeLowercase{\textit{et al.}}: Bare Demo of IEEEtran.cls for IEEE Journals}



\maketitle

\begin{abstract}
Given a social network with nonuniform selection cost of the users, the problem of \textit{Budgeted Influence Maximization} (BIM in short) asks for selecting a subset of the nodes within an allocated budget for initial activation, such that due to the cascading effect, influence in the network is maximized. In this paper, we study this problem with a variation, where a set of nodes are designated as target nodes, each of them is assigned with a benefit value, that can be earned by influencing them, and our goal is to maximize the earned benefit by initially activating a set of nodes within the budget. We call this problem as the \textsc{Earned Benefit Maximization Problem}. First, we show that this problem is NP\mbox{-}Hard and the benefit function is \textit{monotone}, \textit{sub\mbox{-}modular} under the \textit{Independent Cascade Model} of diffusion. We propose an incremental greedy strategy for this problem and show, with minor modification it gives $(1-\frac{1}{\sqrt{e}})$\mbox{-}factor approximation guarantee on the earned benefit. Next, by exploiting the sub\mbox{-}modularity property of the benefit function, we improve the efficiency of the proposed greedy algorithm. Then, we propose a hop\mbox{-}based heuristic method, which works based on the computation of the `expected earned benefit' of the effective neighbors corresponding to the target nodes. Finally, we perform a series of extensive experiments with four real\mbox{-}life, publicly available social network datasets. From the experiments, we observe that the seed sets selected by the proposed algorithms can achieve more benefit compared to many existing  methods. Particularly, the hop\mbox{-}based approach is found to be more efficient than the other ones for solving this problem.
\end{abstract}


% \fi
\IEEEpeerreviewmaketitle

\section{Introduction}
\textit{Social Networks} are an interconnected structure among a group of agents \cite{abraham2009computational}. These are effective platforms, where word\mbox{-}of\mbox{-}mouth effect happens at a large scale and \textit{information}, \textit{ideas}, \textit{rumors} etc.  disseminates widely and rapidly \cite{chierichetti2011rumor}, \cite{kleinberg2008convergence}. This phenomenon has been exploited by the E\mbox{-}Commerce houses for promoting their brands among people \cite{dinh2014cost}, \cite{bagherjeiran2008combining}. The key problem is that which users initially to choose for initiating the diffusion process such that the influence in the network gets maximized. Formally this problem is called as the \emph{Social Influence Maximization Problem} \cite{kempe2003maximizing}. Due to the wider applications in different domains such as \emph{viral marketing} \cite{chen2010scalable}, \textit{social recommendation} \cite{ye2012exploring}, \textit{market basket analysis} \cite{monteserin2018influence}, \textit{prediction of hot topics} \cite{jiang2018predicting} this problem has been studied in different variations. Please, look into \cite{li2018influence} \cite{banerjee2018survey} for recent surveys. The social influence happens due to the cascading process in the underlaying network \cite{li2012diffusion, guille2013information}, and this has huge impact, because human decisions from personal (which place to visit and which restaurant to explore? ) to political (which political party to vote in the coming election?) are influenced by their neighbors, at least to some extent. To study the diffusion process in a social network several diffusion models have been studied. Please look into \cite{li2017survey} for recent survey.  
 \par One of the recently introduced variants of the SIM Problem is the problem of \textit{Budgeted Influence Maximization}  \cite{nguyen2013budgeted}. This problem assumes that users of the network have nonuniform selection cost, which signifies the amount of incentive need to be paid if a user is selected as a seed node. A fixed amount of budget is allocated for the seed set selection process, and the job is to choose highly influential seed nodes within the budget to maximize the influence. There are a few solution methodologies available in the literature for this problem, such as \textit{directed acyclic graph}\mbox{-}based heuristic by Nguyen et al. \cite{nguyen2013budgeted}, \textit{Sample Average Aggregate Scheme} by Guney et al. \cite{guney2017optimal}, ComBIM by Banerjee et al. \cite{banerjee2019combim}. In all these studies, it is implicitly assumed that the influencing each user is equally important, though commercial campaigns are targated in nature, which means a specific brand is to be advertised towards a specific set of users. Because, advertising a brand towards a set of people who do not have any interest towards it does not make any sence. On the other hand, there are some studies in the literature that consider the target user in the influence maximization process \cite{li2015real} \cite{mochalova2014targeted} \cite{nguyen2016cost}. To the best of our knowledge, none of the targeted influence maximization studies considers non\mbox{-}uniform selection cost of the users. 
 \par In target advertisement scenarios, influencing different target user leads to the different amount of benefit and from the advertisers perspective, the main goal is to maximize the total earned benefit. Motivated by this practical scenario, we study the \textsc{Earned Benefit Maximization Problem} (EBM Problem), where the target users are associated with a benefit value, users are associated with a selection cost, and a fixed budget is given. The goal is to select a seed set within the budget to maximize the earned benefit. There are previous studies on this problem by us. In \cite{banerjee2019maximizing}, we came up with a integer programming formulation for this problem. In \cite{banerjee2018priority, banerjee2019maximizin}, we proposed a ranking approach with the exploitation of community structure for this problem. However, in this study our approach to this problem is very different. We start by studying the properties of the earned benefit function and propose a number of solutions followed by experiments. 
Particularly, we make the following contributions in this paper.
 \begin{itemize}
 \item We extend the BIM Problem by considering the notion of target users with non\mbox{-}uniform benefit values and propose the `Earned Benefit Maximization Problem'.
 \item For the EBM Problem, we propose an `incremental greedy strategy' and show with minor modification, this methodology leads to $(1- \frac{1}{\sqrt{e}})$ factor approximation gurrantee on the earned benefit.
 \item We show that the benefit function is \textit{monotone} and \textit{sub\mbox{-}modular} under IC Model of diffusion and exploit this property to improve the efficiency of the incremental greedy algorithm.
 \item Using the concept of `expected earned benefit' of a node, we propose an efficient hop\mbox{-}based heuristic for solving this problem.
 \item We conduct a set of extensive experiments with four real\mbox{-}world publicly available social network datasets for showing the effectiveness and efficiency of the proposed methodologies.
 \end{itemize}
 Rest of the paper is organized as follows. Section \ref{Sec:RW} describes some of the recent studies from the  literature. Section \ref{Sec:BPD} presents preliminary definitions, formally defines the EBM Problem and state its hardness result. The proposed methodologies for solving this problem have been described in Section \ref{Sec:PM}. Section \ref{Sec:EE} contains the experimental evaluations of the proposed methodologies, and finally, in Section \ref{Sec:CFD}, we conclude this study and give the future directions.

\section{Related Work} \label{Sec:RW}
In this section, we present some closely related studies from the literature. This study closely related with the SIM Problem and its variants and more particularly for the targeted users.
\paragraph{Social Influence Maximization and Its Variants} Given a social network of users, which nodes should be chosen for initially injecting the information that causes the maximum influence in the network? This problem is known as the social influence maximization. Initially, this problem was identified in the context of viral marketing  by Damingos and Rechardson \cite{richardson2002mining}. However, Kempe et al. \cite{kempe2003maximizing} were the first to investigate the computational issues of this Problem and proved that it is NP\mbox{-}Hard and proposed an \textit{incremental greedy algorithm}, which admits $(1-\frac{1}{e})$\mbox{-}factor approximation ratio. Their study triggers a vast amount of research on the SIM Problem and hence, a plenty of solution methodologies are available in the literature, such as \textit{Cost\mbox{-}Effective Lazy Forward (CELF)} \cite{leskovec2007cost}, \textit{CELF++} \cite{goyal2011celf++}, \textit{SIMPATH} \cite{goyal2011simpath}, \textit{Two\mbox{-}Phase Influence Maximization} (\textit{$TIM$}) \cite{tang2014influence}, \textit{Influence Maximization Via Martingales} (\textit{IMM}) \cite{tang2015influence}, \textit{Influence Ranking and Influence estimation} (\textit{IRIE}) \cite{jung2012irie}, different \emph{community\mbox{-}based solution methodologies}  \cite{shang2017cofim} \cite{li2018community}, different non\mbox{-}traditional optimization algorithms such as \textit{genetic algorithm} \cite{zhang2017maximizing}, \textit{discrete particle swarm optimization} \cite{wang2017discrete} and many more. Also, there are several variants of this problem studied in the literature, such as $\lambda$\mbox{-}coverage problem \cite{narayanam2011shapley}, budgeted influence maximization problem \cite{nguyen2013budgeted, banerjee2019combim} and many more.
%Cicalese et al. \cite{cicalese2014latency} introduced three different variants of TSS Problem namely \textit{$(\lambda, \beta, \alpha)$ TSS Problem},  \textit{$(\lambda, \beta, \mathcal{A})$ TSS Problem}, and \textit{$(\lambda, \mathcal{A})$ TSS Problem}. Nguyen et al. \cite{nguyen2017social} studied the \textit{influence spectrum problem}, where the goal is to select seed set for multiple times and in each time of various sizes. Instead of applying the same algorithm multiple times, they proposed an efficient updating policy for a precomputed seed set in subsequent times. Sun et al. \cite{sun2018multi} studied the \textit{multi\mbox{-}round influence maximization}, where the goal is to select the seed set across several rounds. They studied the problem in both `adaptive' and `non\mbox{-}adaptive' settings.
\paragraph{Social Influence Maximization for the Targeted Users} Recently, the problem of influence maximization for the targeted users in the social network has been addressed by the researchers. Li et al. \cite{li2015real} studied this problem, and they considered target users, who are relevant to a particular keyword. There solution methodology for this problem was based on the construction of \textit{reverse influence set} and its indexing. Song et al. \cite{song2016targeted} addressed the targeted influence maximization problem by considering the geographical location of the users and the time deadline within which the users should be influenced. Recently, Wen et al. \cite{wen2018maximizing} studied this problem focusing on mainly two issues: how to capture the social influence among the target user and develop an  efficient scheme that can offer the wider influence spread among the target users. Wang et al. \cite{wang2018targeted} also solved the same problem by  considering the impact of budget on the influence spread and incorporating  efficient sampling techniques. 
\par However, to the best of the authors' knowledge, none of the existing studies on targeted influence maximization problem considers the nonuniform benefit associated with each target users and the nonuniform selection cost of users. In this paper, we study the \textsc{Earned Benefit Maximization Problem}, where the target users are associated with non\mbox{-}uniform benefit value and non\mbox{-}uniform selection cost of the users.

%\paragraph{Computational Advertisement in Social Networking Environment}
% Recently, in the literature, the problem of target advertisement in on\mbox{-}line social network environment has been addressed by many researchers. Alon et al. \cite{alon2012optimizing} addressed the budget allocation among the channels and influential customers. Chalermsook et al.  \cite{chalermsook2015social} studied the revenue maximization problem dealing with the multiple advertisers. Abbassi et al. \cite{abbassi2015optimizing} solved the \textit{cost\mbox{-}per\mbox{-}mile} model for display advertisement. Recently, Aslay et al. \cite{aslay2017revenue} addressed the revenue maximization problem in an incentivized social network framework.

\section{Background and Problem Definition} \label{Sec:BPD}
 We consider that the social network is represented as a weighted graph $\mathcal{G}(V, E, \mathcal{P})$, where the \textit{vertex set}  $V(\mathcal{G})$ is the set of users of the network and the \textit{edge set}, $E(\mathcal{G})$ represents the set of \textit{social ties} among the users.  $\mathcal{P}$ is the edge weight function that assigns each edge to its influence probability, i.e., $\mathcal{P}:E(\mathcal{G})\longrightarrow (0,1]$. For any edge $(u_iu_j)\in E(\mathcal{G})$, we denote its influence probability as $\mathcal{P}_{u_i \rightarrow u_j}$. This signifies the probability that the user $u_i$ will be able to influence $u_j$. We denote the number of nodes and edges of $\mathcal{G}$ by $n$ and $m$, respectively. Next, we briefly describe the Independent Cascade Model, which we consider as the underlying diffusion model in our study. 
\subsection{The Independent Cascade Model}
The Independent Cascade Model is one of the models, which has been predominantly used in influence maximization literature \cite{kempe2003maximizing, tang2014influence, tang2015influence}. Here, the diffusion of information starts from a set of nodes selected initially and known as the \textit{seed nodes}. All the nodes of the network are ignorant of the information and the seed nodes are informed at time $t=0$. Now, from these seed nodes, information is diffused by the following rules:
\begin{itemize}
\item information is diffused in discrete time steps,
\item a node can be either one of the two states: `active' (`influenced') or inactive (`uninfluenced'),
\item a node can change its state from inactive to active however, not the vice-versa. 
\item once a node is influenced, it will remain in this state.
\end{itemize}
 Each active node (say, $u_i$) at current time stamp (say $t$) will get a chance to activate its currently \textit{inactive} neighbors ($u_j \in \mathcal{N}(u_i)$ and $u_j$ is inactive) with probability as their edge weight. If any one of them succeeds, then $u_j$ will become an active node at time $t+1$. Only the recently active node can take part in the triggering process. This process stops, when no more node activation is possible. Next, we introduce the Earned Benefit Maximization Problem.
 \subsection{The Earned Benefit Maximization Problem}
In this problem, along with the social network $\mathcal{G}(V, E, \mathcal{P})$, a subset of the users $\mathcal{D}$ is given as the target users. Each of them is associated with a benefit, which can be earned by influencing the corresponding target user. This can be characterized by the \emph{benefit function} $b: \mathcal{D}\longrightarrow \mathbb{R}^{+}$. For any $u \in \mathcal{D}$ his benefit is denoted as $b(u)$ and for any $u \notin \mathcal{D}$, $b(u)=0$. 
%and for a subset of target user $\mathcal{D} \subset T$, the total benefit by the subset target nodes is denoted by $\beta(\mathcal{D})$, where $\beta(\mathcal{D})=\underset{u \in \mathcal{D}}{\sum}\beta(u)$. 
 For a seed set $\mathcal{S}$, the set of nodes influenced by it is denoted by $I(\mathcal{S})$. As the diffusion of information under IC Model is a probabilistic process, the influence of a seed set is measured in terms of expectation. Hence, the number of influenced nodes due to the seed set $\mathcal{S}$ is $\sigma(\mathcal{S})=\mathbb{E}[\vert I(\mathcal{S})\vert]$, where, $\sigma(.)$ is the \emph{social influence function} \cite{kempe2003maximizing}. Now, the \emph{earned benefit} by the seed set $\mathcal{S}$ is defined as $\beta(\mathcal{S})=\underset{u \in \mathcal{D} \cap I(\mathcal{S})}{\sum}b(u)$. Here, $\beta(.)$ is the \emph{earned benefit function}, that maps each subset of the nodes to is expected earned benefit value, i.e., $\beta:2^{V(\mathcal{G})} \longrightarrow \mathbb{R}_{\geq 0}$.
 \par In real\mbox{-}world campaigns, earned benefit maximization is done by conducting an information diffusion process. As the real\mbox{-}life social networks are formed by the rational human agents, if a user is selected as seed, incentivization is required. This can be characterized by the \emph{cost function} $\mathcal{C}: V(\mathcal{G}) \longrightarrow \mathbb{R}^{+}$. Selection cost associated with the node $u$ is denoted as $\mathcal{C}(u)$. For a subset of nodes $\mathcal{S}$, their selection cost is denoted as $\mathcal{C}(\mathcal{S})=\underset{u \in \mathcal{S}}{\sum}\mathcal{C}(u)$, and a fixed amount of budget $\mathcal{B}$ is allocated for seed set selection. Hence, the problem here is to choose a subset $\mathcal{S}$ from $V(\mathcal{G})$ to maximize the function $\beta(\mathcal{S})$ subject to the constraint $\mathcal{C}(\mathcal{S})\leq \mathcal{B}$. Formally, the problem can be expressed as follows:
\begin{tcolorbox}

\underline{\textsc{Earned Benefit Maximization Problem}} \\
%\vspace*{2 cm} 
\textbf{Input:} Social Network $\mathcal{G}(V, E, \mathcal{P})$, Target Nodes $\mathcal{D}$, Cost Function $\mathcal{C}$, Benefit Function $b$, and Budget $\mathcal{B}$.

\textbf{Problem:} Find out the seed set ($\mathcal{S}$) such that $\underset{u \in \mathcal{S}}{\sum} \mathcal{C}(u) \leq \mathcal{B}$ and for any other seed set $\mathcal{S}^{'}$ with $\underset{v \in \mathcal{S}^{'}}{\sum} \mathcal{C}(v) \leq \mathcal{B}$, $\beta(\mathcal{S})\geq \beta(\mathcal{S}^{'})$.
\end{tcolorbox} 
 
 The EBM Problem is basically the generalization of the BIM Problem, which is NP\mbox{-}Hard under the IC Model of diffusion \cite{nguyen2013budgeted}. Hence, Theorem \ref{Th:1} holds.
\begin{mythm} \label{Th:1}
The EBM Problem is NP\mbox{-}Hard under Independent Cascade Model of diffusion. 
\end{mythm}
This result motivates us to design suitable approximation algorithm and heuristic solution for this problem. We discuss them in the next Section.
\section{Proposed Methodology} \label{Sec:PM}
In this section, we present our proposed methodologies for the EBM  problem. Prior to that, we establish two properties of the benefit function, which will be used subsequently.
\subsection{Properties of the Benefit Function}
As mentioned previously, the benefit earned by a given seed set $\mathcal{S}$ is defined as $\beta(\mathcal{S})=\underset{u \in \mathcal{D} \cap I(\mathcal{S})}{\sum} b(u)$. So, the benefit function can be thought of a set function, which is defined on the ground set $V(\mathcal{G})$, i.e., $\beta: 2^{V(\mathcal{G})} \longrightarrow \mathbb{R}^{+}$. Now, we prove two important properties of the benefit function, namely \textit{monotonicity} and \textit{sub\mbox{-}modularity}. This two properties are exploited for proving the approximation guarantee of Algorithm \ref{Algo:1a}.
 \begin{mydef}[Non\mbox{-}negativity and Monotonicity of Set Function]
 A set function $f(.)$ defined over the ground set $V(\mathcal{G})$ is said to be non\mbox{-}negative if $\forall \mathcal{S} \subseteq V(\mathcal{G})$, $f(\mathcal{S}) \geq 0$ and monotone if $ \forall \mathcal{S} \subseteq \mathcal{T}$, $f( \mathcal{S}) \leq f( \mathcal{T})$.
 \end{mydef}
 \begin{mylem}
 The benefit function, $\beta(.)$ is non\mbox{-}negative and monotone under IC Model of diffusion.
 \end{mylem}
\begin{proof}
 It is reported in the literature that the social influence function is non\mbox{-}negative and monotone \cite{kempe2003maximizing}. As $\forall u \in V(\mathcal{G})$, $b(u) \geq 0$, it is trivial to observe that $\forall S\subseteq V(\mathcal{G})$, $\beta(\mathcal{S}) \geq 0$. By the monotonicity property of $\sigma(.)$, $\forall \mathcal{S} \subseteq \mathcal{S}^{'}$,
 \begin{center}
 $I(\mathcal{S}) \subseteq I(\mathcal{S}^{'})$\\
 $\Rightarrow \mathcal{D} \cap I(\mathcal{S}) \subseteq \mathcal{D} \cap I(\mathcal{S}^{'})$ \\
 $\Rightarrow \underset{u \in \mathcal{D} \cap I(\mathcal{S})}{\sum} b(u)\leq \underset{u \in \mathcal{D} \cap I(\mathcal{S}^{'})}{\sum} b(u)$ \\
 $\Rightarrow \beta(\mathcal{S}) \leq \beta(\mathcal{S}^{'})$
 \end{center}
 which means $\beta(.)$ is monotone. This completes the proof.
 \end{proof}
 \begin{mydef}[Sub\mbox{-}modularity of Set Function]
 A set function $f(.)$ defined over the ground set $V(\mathcal{G})$ is said to be sub\mbox{-}modular if $\forall \mathcal{S} \subseteq \mathcal{T} \subset V(\mathcal{G})$ and $\forall u \in  V(\mathcal{G}) \setminus \mathcal{T}$, the following condition is met:
 \begin{equation}
 f(\mathcal{S} \cup \{u\}) - f(\mathcal{S}) \geq f(\mathcal{T} \cup \{u\}) - f(\mathcal{T})
 \end{equation}
 \end{mydef}
 \begin{mylem} \label{Lemma:Submodularity}
 The benefit function, $\beta(.)$ is sub\mbox{-}modular under IC Model of diffusion.
 \end{mylem}
\begin{proof}
 In the literature, it is mentioned that the social influence function, $\sigma(.)$ is sub\mbox{-}modular under the IC Model of diffusion \cite{kempe2003maximizing}. Let us assume that $\mathcal{S} \subseteq \mathcal{T} \subset V(\mathcal{G})$ and $u \in V(\mathcal{G}) \setminus \mathcal{S}$. Now, from the definition of $\beta(.)$, we have
 \begin{center}
 $\beta(\mathcal{S} \cup \{u\})-\beta(\mathcal{S})=\underset{u \in \mathcal{D} \cap I(\mathcal{S} \cup \{u\})}{\sum} b(u) - \underset{u \in \mathcal{D} \cap I(\mathcal{S})}{\sum} b(u)$
 \end{center}
 By simple set theoretic interpretation, we can write 
 \begin{center}
 $I(\mathcal{S} \cup \{u\})= (I(\mathcal{S}) \cup I( \{u\})) \setminus (I(\mathcal{S}) \cap I(\{u\}))$\\
 $\Rightarrow \mathcal{D} \cap I(\mathcal{S} \cup \{u\}) = (\mathcal{D} \cap I(\mathcal{S})) \cup (\mathcal{D} \cap I(\{u\})) \setminus (\mathcal{D} \cap (I(\mathcal{S}) \cap I(\{u\})))$
 \end{center}
 Hence, 
 \begin{center}
 $\beta(\mathcal{S} \cup \{u\})-\beta(\mathcal{S})= \underset{u \in \mathcal{D} \cap I(\mathcal{S})}{\sum} b(u) + \underset{u \in \mathcal{D} \cap I(\{u\})}{\sum} b(u) - \underset{u \in \mathcal{D} \cap (I(\mathcal{S}) \cap I(\{u\}))}{\sum} b(u) - \underset{u \in \mathcal{D} \cap I(\mathcal{S})}{\sum} b(u)$\\
 \vspace{0.3 cm}
 $\Rightarrow \beta(\mathcal{S} \cup \{u\})-\beta(\mathcal{S})= \underset{u \in \mathcal{D} \cap I(\{u\})}{\sum} b(u) - \underset{u \in \mathcal{D} \cap (I(\mathcal{S}) \cap I(\{u\}))}{\sum} b(u)$\\
\vspace{0.3 cm}
$\geq \underset{u \in \mathcal{D} \cap I(\{u\})}{\sum} b(u) - \underset{u \in \mathcal{D} \cap (I(\mathcal{T}) \cap I(\{u\}))}{\sum} b(u)$\\
\vspace{0.3 cm}
 [This is due to to the monotonicity property of $I(.)$]
 \vspace{0.3 cm}
 $=\underset{u \in \mathcal{D} \cap I(\mathcal{T} \cup \{u\})}{\sum} b(u) - \underset{u \in \mathcal{D} \cap I(\mathcal{T})}{\sum} b(u)$ \\
 \vspace{0.3 cm}
 $=\beta(\mathcal{T} \cup \{u\}) - \beta(\mathcal{T})$
 \end{center}
 We obtain the inequality required to show sub\mbox{-}modularity property of $\beta(.)$. This completes the proof.  
 \end{proof}
\subsection{Incremental Greedy Algorithm}
 Let us assume $\mathcal{S}$ be the seed set and $u \in V(\mathcal{G})\setminus \mathcal{S}$. We define the \textit{marginal gain in benefit} for the node $u$ with respect to the seed set $\mathcal{S}$ as the amount of increased benefit when the node $u$ is included in the seed set $\mathcal{S}$. Formally, it is stated in Definition \ref{Def:1}.
\begin{mydef} [Marginal Gain in Earned Benefit] \label{Def:1}
Given a seed set $\mathcal{S}$ and a node $u$, which is currently not in the seed set, i.e., $u \in V(\mathcal{G}) \setminus \mathcal{S}$, its marginal gain in the earned benefit with respect to the seed set $\mathcal{S}$ is denoted as $\Delta_{\beta}(\mathcal{S}\vert u)$ and defined as 
\begin{equation}
\Delta_{\beta}(\mathcal{S}\vert u)=\beta(\mathcal{S} \cup \{u\}) - \beta(\mathcal{S})
\end{equation}
\end{mydef}
The working principle of the proposed incremental greedy strategy is as follows. Starting with an empty seed set, this procedure incrementally selects a node within the budget that causes the maximum marginal gain per unit cost. Let $\mathcal{S}^{i}$ and $\mathcal{B}^{i}$ denote the seed set and remaining budget at the end of $i$\mbox{-}th iteration. In the $(i+1)$\mbox{-}th iteration, the node $u$ is added in the seed set $\mathcal{S}^{i}$, i.e., $\mathcal{S}^{i+1}=\mathcal{S}^{i} \cup \{u\}$, if the following condition is met.
\begin{equation}
u=\underset{v \in V(\mathcal{G}) \setminus \mathcal{S}^{i}, \mathcal{C}(v) \leq \mathcal{B}^{i}}{argmax} \frac{\Delta_{\beta}(\mathcal{S}\vert v)}{\mathcal{C}(v)}
\end{equation}  
 In an iteration, if the no seed node is selected within the remaining budget, then $u$ is null and if this happens, then the procedure is exiting. Algorithm \ref{Algo:1} states the procedure.
\begin{algorithm}[h]
 \caption{Incremental Greedy Algorithm for the EBM Problem }
 \label{Algo:1}
 \begin{algorithmic}[1]
 \renewcommand{\algorithmicrequire}{\textbf{Input:}}
 \renewcommand{\algorithmicensure}{\textbf{Output:} }
 \REQUIRE Social Network $\mathcal{G}(V, E, \mathcal{P})$, Target Nodes $\mathcal{D}$, Cost Function $\mathcal{C}$, Benefit Function $b$, and Budget $\mathcal{B}$. 
 \ENSURE   The seed set $\mathcal{S} \subseteq V(\mathcal{G})$ such that $\underset{u \in \mathcal{S}}{\sum} \mathcal{C}(u) \leq \mathcal{B}$. 
  \STATE $\mathcal{S} \leftarrow \phi$\;
   \WHILE{$\mathcal{B}>0$}
   \STATE $u\leftarrow \underset{v \in V(\mathcal{G})\setminus \mathcal{S}, \mathcal{C}(v) \leq \mathcal{B}}{argmax} \frac{\Delta_{\beta}(\mathcal{S}\vert v)}{\mathcal{C}(v)}$\;
   \IF{$u=\phi$}
	\STATE	$break$\;
	\ENDIF
\STATE	$\mathcal{S}\leftarrow \mathcal{S} \cup \{u\}$\;
 \STATE       $\mathcal{B} \leftarrow \mathcal{B}-\mathcal{C}(u)$\;
   \ENDWHILE
\STATE $return \ \mathcal{S}$\;
 \end{algorithmic} 
\end{algorithm}
Though the Algorithm \ref{Algo:1} is simple to understand, it does not give any bounded approximation guarantee on the earned benefit and we demonstrate this claim with an example.
\begin{myexam}
Let us assume, a network with $p+1$ nodes $V(\mathcal{G})=\{u, v_1, v_2, \dots, v_p\}$, where $u$ is an isolated node and the remaining nodes connected within themselves with each edge having the diffusion probability $1$. The entire vertex set of the network is the target node set, i.e., $\mathcal{D}=V(\mathcal{G})$. Benefit associated with each target node is $1$. For each $v_i$, its associated selection cost is $p$ and the selection cost of $u$ is $(1-\epsilon)$, where $0<\epsilon <1$. The allocated budget for the seed set selection process is $p$. The optimal algorithm for this problem should select any $v_i$ node and achieve the earned benefit of amount $p$ by influencing all the remaining $v_i$ nodes. However, as Algorithm \ref{Algo:1} selects the seed node based on the marginal gain in the earned benefit per unit cost, it will  select the node $u$ and not any $v_i$. For the node $u$, the value of $\frac{\Delta_{\beta}(\mathcal{S}\vert u)}{\mathcal{C}(u)}$, when $\mathcal{S}=\phi$, is $\frac{1}{1-\epsilon}$. On the other hand, for any $v_i$, the value of $\frac{\Delta_{\beta}(\mathcal{S}\vert v_i)}{\mathcal{C}(v_i)}$ is $1$. As $\frac{1}{1-\epsilon} > 1$, Algorithm \ref{Algo:1} selects the node $u$. After selecting the node $u$, the remaining budget will be $p+\epsilon-1$, which is less than $p$. Within this budget none of the $v_i$ nodes can be selected as each of them has the selection cost $p$. Hence, Algorithm \ref{Algo:1} terminates by earning the benefit $1$ and returning an  unutilized budget of amount $p+\epsilon-1$. The approximation ratio of the Algorithm \ref{Algo:1} is define as
\begin{center}
$Ratio_{Algo \ref{Algo:1}}= \frac{\text{Benefit earned by the seed set selected by Algorithm \ref{Algo:1}}}{ \text{Benefit earned by the optimal seed set}}$
\end{center}
In this example the value of $Ratio_{Algo \ \ref{Algo:1}}$ is $\frac{1}{p}$. If the value of $p$ is arbitrarily large, then the approximation ratio of Algorithm \ref{Algo:1} becomes very very less. Hence, Theorem \ref{Thm:1} holds.
\end{myexam}
\begin{mythm} \label{Thm:1}
Algorithm \ref{Algo:1} does not provide any constant approximation guarantee.
\end{mythm} 
Now, we present two important inequities on the iterative performance of Algorithm \ref{Algo:1} and this result will be used subsequently.
\begin{mylem} \label{Lemma:seq}
After each iteration of the `while' loop $i=1,2, \dots, p+1$, the following inequality always holds
\begin{equation} \label{Eq:imp}
\beta(\mathcal{S}_{i})-\beta(\mathcal{S}_{i-1}) \geq \frac{\mathcal{C}(u_i)}{\mathcal{B}}(\beta(\mathcal{S}^{opt})-\beta(\mathcal{S}_{i-1})).
\end{equation}
\end{mylem}
\begin{proof}
The value of $\beta(\mathcal{S}^{opt})-\beta(\mathcal{S}_{i-1})$ is no more than the sum of the benefit values of the target nodes that are influenced by the seed nodes in $\mathcal{S}^{opt}$, however, not by the nodes in $\mathcal{S}_{i-1}$. For each node in $\mathcal{S}^{opt} \setminus \mathcal{S}_{i-1}$, the earned benefit to cost ratio could be at most $\frac{B_i}{\mathcal{C}(u_i)}$, where $B_i$ is the earned benefit by the nodes in $\mathcal{S}_{i}$ but not in $\mathcal{S}_{i-1}$. This is because $\mathcal{S}_{i}$ maximizes this ratio in Algorithm \ref{Algo:1}. Since the total selection cost of the nodes in $\mathcal{S}^{opt} \setminus \mathcal{S}_{i-1}$ is bounded by the budget $\mathcal{B}$, the total earned benefit due to the target nodes in $I(\mathcal{S}^{opt}) \setminus I(\mathcal{S}_{i-1})$ can be at most $\mathcal{B}\frac{B_i}{\mathcal{C}(u_i)}$. Hence, we have
\begin{equation} \label{Eq:5}
\beta(\mathcal{S}^{opt})-\beta(\mathcal{S}_{i-1}) \leq \mathcal{B} \frac{b_i}{\mathcal{C}(u_i)}
\end{equation}
By definition, we have
\begin{equation} \label{Eq:6}
B_i=\beta(\mathcal{S}_{i})-\beta(\mathcal{S}_{i-1})
\end{equation} 
From the Equations (\ref{Eq:5}) and (\ref{Eq:6}), we have
\begin{center}
$\beta(\mathcal{S}^{opt})-\beta(\mathcal{S}_{i-1}) \leq \mathcal{B} \frac{\beta(\mathcal{S}_{i})-\beta(\mathcal{S}_{i-1})}{\mathcal{C}(u_i)}$\\
\vspace{0.2 cm}
$\Rightarrow \beta(\mathcal{S}_{i})-\beta(\mathcal{S}_{i-1}) \geq  \frac{\mathcal{C}(u_i)}{\mathcal{B}} (\beta(\mathcal{S}^{opt})-\beta(\mathcal{S}_{i-1}))$
\end{center} 
This completes the proof.
\end{proof}
\begin{mylem} \label{Lemma:3}
In any arbitrary iteration $i=1,2, \dots, p+1$ of the While loop from Line $2$ to $9$ of Algorithm \ref{Algo:1}, the following condition will be true 
\begin{center}
$\beta(\mathcal{S}_{i}) \geq [1- \underset{i}{\prod}(1- \frac{\mathcal{C}(u_i)}{\mathcal{B}})] \beta(\mathcal{S}^{opt})$
\end{center} 
\end{mylem} 
\begin{proof}
We prove this statement by the method of induction on the iteration of the `while' loop. For the first iteration, i.e., at $i=1$, we need to show,
\begin{center}
$\beta(\mathcal{S}_{1}) \geq \frac{\mathcal{C}(u_1)}{\mathcal{B}} \beta(\mathcal{S}^{opt})$
\end{center}
From Lemma \ref{Lemma:seq}, by putting $i=0$ in Equation (\ref{Eq:imp}) we have,
\begin{center}
$\beta(\mathcal{S}_{1})-\beta(\mathcal{S}_{0}) \geq \frac{\mathcal{C}(u_1)}{\mathcal{B}}(\beta(\mathcal{S}^{opt})-\beta(\mathcal{S}_{0}))$
\end{center}
As we are starting with an empty seed set, hence $\mathcal{S}_{0} = \emptyset$ and $\beta(\mathcal{S}_{0})=0$. This clearly implies that $\beta(\mathcal{S}_{1}) \geq \frac{\mathcal{C}(u_1)}{\mathcal{B}} \beta(\mathcal{S}^{opt})$.
\par Now, suppose the statement holds till $(i-1)^{th}$ iteration. We show that the statement holds in the $i^{th}$ iteration as well. Now, 
\begin{center}
$\beta(\mathcal{S}_{i})=\beta(\mathcal{S}_{i-1})+(\beta(\mathcal{S}_{i}) - \beta(\mathcal{S}_{i-1}))$ \\
\vspace{0.2 cm}
$ \geq \beta(\mathcal{S}_{i-1}) + \frac{\mathcal{C}(u_i)}{\mathcal{B}} (\beta(\mathcal{S}^{opt})-\beta(\mathcal{S}_{i-1}))$ \\
\vspace{0.2cm}
$ = (1-\frac{\mathcal{C}(u_i)}{\mathcal{B}}) \beta (\mathcal{S}_{i-1}) + \frac{\mathcal{C}(u_i)}{\mathcal{B}} \beta(\mathcal{S}^{opt})$ \\
\vspace{0.2 cm}
$\geq(1- \frac{\mathcal{C}(u_i)}{\mathcal{B}}) (1- \prod_{k=1}^{i-1}(1- \frac{\mathcal{C}(u_k)}{\mathcal{B}})) \beta(\mathcal{S}^{opt}) + \frac{\mathcal{C}(u_i)}{\mathcal{B}} \beta(\mathcal{S}^{opt})$ \\
\vspace{0.2 cm}
$ \geq [1-\prod_{k=1}^{i} (1-\frac{\mathcal{C}(u_k)}{\mathcal{B}})] \beta(\mathcal{S}^{opt})$ 
\end{center}
Here, the first inequality is due to Lemma \ref{Lemma:seq} and the second one is due to inductive hypothesis.   
\end{proof} 
\par Algorithm \ref{Algo:1} can be modified for yielding a constant approximation ratio on the earned benefit. Let $\mathcal{S}_{G}$ be the seed set generated by the Algorithm \ref{Algo:1}. $u_{max}$ be the node that has the highest individual benefit gain. We compare the earned benefit, when the seed set is $\mathcal{S}_{G}$ and the node is $u_{max}$. We return the seed set that maximizes the earned benefit. Algorithm \ref{Algo:1a} formally states the procedure.
\begin{algorithm}[H]
 \caption{Modified Incremental Greedy Algorithm }
 \label{Algo:1a}
 \begin{algorithmic}[1]
 \renewcommand{\algorithmicrequire}{\textbf{Input:}}
 \renewcommand{\algorithmicensure}{\textbf{Output:} }
 \REQUIRE Social Network $\mathcal{G}(V, E, \mathcal{P})$, Target Nodes $\mathcal{D}$, Cost Function $\mathcal{C}$, Benefit Function $b$, and Budget $\mathcal{B}$. 
 \ENSURE   The seed set $\mathcal{S} \subseteq V(\mathcal{G})$ such that $\underset{u \in \mathcal{S}}{\sum} \mathcal{C}(u) \leq \mathcal{B}$.  
  \STATE $\mathcal{S} \leftarrow \phi$\;
   \STATE $\mathcal{S}_{G}= \text{Seed Set selected by Algorithm } \ref{Algo:1}$\;
   \STATE $u_{max}= \underset{v \in V(\mathcal{G}), \mathcal{C}(v) \leq \mathcal{B}}{argmax} \beta(v)$\;
   \STATE $\mathcal{S}=\underset{<\mathcal{S}_{G}, u_{max}>}{argmax}(\beta(\mathcal{S}_{G}), \beta(u_{max}))$\;
\STATE $return \ \mathcal{S}$\;
 \end{algorithmic} 
\end{algorithm}
Algorithm \ref{Algo:1a} provides bounded approximation guarantee, which is stated in Theorem \ref{Thm:2}. 
\begin{mythm} \label{Thm:2}
$\mathcal{S}^{\mathcal{A}}$ is the seed set selected by Algorithm \ref{Algo:1a} and $\mathcal{S}^{opt}$ be the optimal seed set, then $\beta(\mathcal{S}^{\mathcal{A}}) \geq (1-\frac{1}{\sqrt{e}}) \beta(\mathcal{S}^{opt})$, where $e=\sum_{x=1}^{\infty} \frac{1}{x!}$. In other words, Algorithm \ref{Algo:1a} provides an approximation guarantee of $(1-\frac{1}{\sqrt{e}})$.
\end{mythm}
\begin{proof}
The strategy of this proof has been used previously for proving the approximation bound of the \textit{Budgeted Maximum Coverage Problem} by Khuller et al. \cite{khuller1999budgeted}. Here, we prove the statement by case\mbox{-}wise analysis of Algorithm \ref{Algo:1a}.\\
\textbf{Case I}\\
If there exists one node $u \in V(\mathcal{G})$, which has the earned benefit $\beta(u)$, and $\beta(u)$ is found to be greater than equal to $\frac{\beta(\mathcal{S}^{opt})}{2}$, then $u$ will be selected as $u_{max}$ in Algorithm \ref{Algo:1a}. In this case, the approximation ratio of Algorithm \ref{Algo:1a} will be as follows:
\begin{center}
$Ratio_{Algo \ref{Algo:1a}}= \frac{\beta(\mathcal{S}^{\mathcal{A}})}{\beta(\mathcal{S}^{opt})}\geq \frac{\beta(\mathcal{S}^{opt})}{2 \beta(\mathcal{S}^{opt})} = \frac{1}{2}$
\end{center} 
\textbf{Case II}\\
If Case I does not happen, then there does not exist any $u \in V(\mathcal{G})$, for which $\beta(u)$ is greater than $\frac{\beta(\mathcal{S}^{opt})}{2}$. This can be divided into two sub\mbox{-}cases.\\
\textbf{Case IIa}\\ 
 Now, if we have  $\mathcal{C}(\mathcal{S}^{\mathcal{A}}) < \frac{\mathcal{B}}{2}$, then $\forall u \in V(\mathcal{G})\setminus \mathcal{S}$, $\mathcal{C}(u)> \frac{\mathcal{B}}{2}$. Hence, no more node can be added to $\mathcal{S}^{\mathcal{A}}$. Otherwise, the budget constraint  will be violated. Without the loss of generality, let us assume that $\mathcal{S}^{opt} \neq \mathcal{S}^{\mathcal{A}}$. In this case, $\mathcal{S}^{opt}$ can contain one extra node without violating the budget constraint. Now, as the function $\beta(.)$ is sub\mbox{-}modular, hence,
\begin{center}
$\beta(\mathcal{S}^{opt} \cap \mathcal{S}^{\mathcal{A}}) + \beta(\{v\}) \geq \beta((\mathcal{S}^{opt} \cap \mathcal{S}^{\mathcal{A}}) \cap \{v\}) + \beta((\mathcal{S}^{opt} \cap \mathcal{S}^{\mathcal{A}}) \cup \{v\})$ \\
\vspace{0.2 cm}
$\geq \beta(\phi) + \beta(\mathcal{S}^{opt})$ \\
\vspace{0.2 cm}
$=\beta(\mathcal{S}^{opt})$
\end{center}
As $\forall u \in  V(\mathcal{G})$, $\beta(u)<\frac{\beta(\mathcal{S}^{opt})}{2}$. This clarifies that $\beta(\mathcal{S}^{opt} \cap \mathcal{S}^{\mathcal{A}}) \geq \frac{\beta(\mathcal{S}^{opt})}{2}$. As $\mathcal{S}^{opt}=\mathcal{S}^{\mathcal{A}} \cup \{v\}$ and $v \notin \mathcal{S}^{\mathcal{A}}$, $\mathcal{S}^{opt} \cap \mathcal{S}^{\mathcal{A}}= \mathcal{S}^{\mathcal{A}}$. This essentially means $\beta(\mathcal{S}^{\mathcal{A}}) \geq \frac{\beta(\mathcal{S}^{opt})}{2}$. In this case, the approximation ratio of the Algorithm \ref{Algo:1a} will be as follows:
\begin{center}
$Ratio_{Algo \ref{Algo:1a}}= \frac{\beta(\mathcal{S}^{\mathcal{A}})}{\beta(\mathcal{S}^{opt})} \geq \frac{\beta(\mathcal{S}^{opt})}{2 \beta(\mathcal{S}^{opt})} \geq \frac{1}{2}$
\end{center}
\textbf{Case IIb}\\
If $\mathcal{C}(\mathcal{S}^{\mathcal{A}}) \geq \frac{\mathcal{B}}{2}$, we first observe that for $n$ real numbers $a_1, \ a_2, \dots, a_n$ and $\sum_{i=1}^{n} a_i \geq \alpha A$, the function $\prod_{i=1}^{n}(1-\frac{a_i}{A})$ attains its maximum value, when $a_i=\frac{\alpha A}{n}$. Hence, by Lemma \ref{Lemma:3}, we have 
\begin{center}
$\beta(\mathcal{S}^{\mathcal{A}}) \geq [1- \prod_{i=1}^{\vert \mathcal{S}^{\mathcal{A}} \vert} (1- \frac{\mathcal{C}(u_i)}{\mathcal{B}})] \beta(\mathcal{S}^{opt})$\\
\vspace{0.2 cm}
$\geq [1-(1-\frac{1}{2i})^{i}] \beta(\mathcal{S}^{opt})$\\
\vspace{0.2 cm}
$\geq (1-\frac{1}{\sqrt{e}}) \beta(\mathcal{S}^{opt})$
\end{center}
Hence, the worst case performance guarantee of Algorithm \ref{Algo:1a} is  $(1-\frac{1}{\sqrt{e}})$. This proves the statement.
\end{proof}
Now, we investigate the time requirement of Algorithms \ref{Algo:1} and \ref{Algo:1a}. For both of them, it is easy to observe that the time requirement is heavily dependent on the earned benefit calculation for a given seed set. It is reported in the literature that counting the number of influenced nodes for a given seed set is $\#P \mbox{-}Hard$ problem \cite{kempe2003maximizing}. With this argument, we can say that for a given seed set $\mathcal{S}$, computing the exact value of the earned benefit is also $\#P \mbox{-}Hard$. Hence, we estimate this value, the way influence of a seed set is estimated \cite{kempe2003maximizing}. First, a number (say $ \mathcal{R} $) of sampled graphs of $\mathcal{G}$, i.e., $G_1, G_2, \ldots, G_{\mathcal{R}}$ are generated, and for all $p \in [\mathcal{R}]$, for all $(u_iu_j) \in E(\mathcal{G})$, $(u_iu_j) \in E(G_{p})$ with probability $\mathcal{P}_{u_i \rightarrow u_j}$ and $(u_iu_j) \notin E(G_{p})$ with probability $(1-\mathcal{P}_{u_i \rightarrow u_j})$. Now, earned benefit is computed in all of the sampled graphs and the average value is returned as its approximate value, which is given in  Equation \ref{Eq:EB} 
\begin{equation} \label{Eq:EB}
\beta_{\mathcal{G}}(\mathcal{S})= \frac{\sum_{p=1}^{|\mathcal{R}|} \beta_{G_{p}}(\mathcal{S})}{\mathcal{R}}
\end{equation} 
  
 \par If $|\mathcal{S}|=k$, then traversing $ \mathcal{R} $ subgraphs will require $\mathcal{O}(k(m+n)  \mathcal{R})$ time. Let, $C_{min}$ be the minimum selection cost among all the users. Maximum number of possible iterations of the While loop (Line 2 to 9) in Algorithm \ref{Algo:1} is $\frac{\mathcal{B}}{C_{min}}$. Hence, $k \leq \frac{\mathcal{B}}{C_{min}}$. In Algorithm \ref{Algo:1}, in each iteration, maximum number of times earned benefit estimations are done is of $\mathcal{O}(n)$. Hence, the total number of times earned benefit estimations are of $\mathcal{O}(\frac{\mathcal{B}}{C_{min}}.n)$. The required computational time for Algorithm \ref{Algo:1} is $\mathcal{O}((\frac{\mathcal{B}}{C_{min}})^{2}.n(m+n). \mathcal{R})$.
 \par In Algorithm \ref{Algo:1a}, along with the incremental greedy strategy, the node, which can grab the maximum earned benefit has to be found out (Line 3 of Algorithm \ref{Algo:1a}). This can be done $\mathcal{O}(n)$ earned benefit estimations with a single seed node, and this will take $\mathcal{O}(n(m+n) \mathcal{R})$ time. Hence, running time of Algorithm \ref{Algo:1a} is of $\mathcal{O}((\frac{\mathcal{B}}{C_{min}})^{2}.n(m+n). \mathcal{R}  + n(m+n) \mathcal{R} ) \approx \mathcal{O}((\frac{\mathcal{B}}{C_{min}})^{2}.n(m+n).\mathcal{R} )$. If we do the on\mbox{-}line sampling of the input social network for sampled graph generation, then only one subgraph is required per iteration. For storing, this network will take $\mathcal{O}(n+m)$ space. Storing the seed set requires $\mathcal{O}(\frac{\mathcal{B}}{C_{min}})$ space. Hence, the total amount of space required by both Algorithms \ref{Algo:1} and \ref{Algo:1a} is $\mathcal{O}(m+n+\frac{\mathcal{B}}{C_{min}})$ and the number of seed nodes is generally found to be much much less than the number of nodes, i.e., $\frac{\mathcal{B}}{C_{min}} << n$. Hence, $\mathcal{O}(m+n+\frac{\mathcal{B}}{C_{min}}) \approx\mathcal{O}(m+n)$. Hence, Theorem \ref{Th:3} holds.
\begin{mythm} \label{Th:3}
 Algorithms \ref{Algo:1} and \ref{Algo:1a} have the running time of $\mathcal{O}((\frac{\mathcal{B}}{C_{min}})^{2}.n(m+n). \mathcal{R})$ and space requirement of $\mathcal{O}(m+n)$.  
\end{mythm}
\subsection{Improving the Efficiency of Algorithm \ref{Algo:1a}}
 Though Algorithm \ref{Algo:1a} provides a provable approximation bound on the earned benefit, it is highly inefficient, as it estimates the earned benefit many times. Here, we present an improvised version of Algorithm \ref{Algo:1a} in Algorithm \ref{Algo:3} by removing redundant earned benefit estimations  due to the exploitation of the sub\mbox{-}modularity property of the earned benefit function.
 \begin{algorithm} [H]
	\caption{Incremental Greedy Algorithm with Improve Performance in terms of Efficiency (IGAIP).}
	\label{Algo:3}
	\begin{algorithmic}[1]

\renewcommand{\algorithmicrequire}{\textbf{Input:}}

\renewcommand{\algorithmicensure}{\textbf{Output:} }

\REQUIRE Social Network $\mathcal{G}(V, E, \mathcal{P})$, Target Nodes $\mathcal{D}$, Cost Function $\mathcal{C}$, Benefit Function $b$, and Budget $\mathcal{B}$.  
% $\mathcal{C}: V(\mathcal{G}) \longrightarrow \mathcal{R}^{\geq 0}$,
% $\mathcal{B}$\alpha
\ENSURE The seed set $\mathcal{S} \subseteq V(\mathcal{G})$ such that $\underset{u \in \mathcal{S}}{\sum} \mathcal{C}(u) \leq \mathcal{B}$. 
%\Parameter{A parameter for the algorithm}
\STATE $\mathcal{S} \longleftarrow \phi$\;
%\STATE $Flag = 1$\;
%\FOR{All $u \in V({\mathcal{G}})$}
%\STATE $\text{Compute } \Delta_{\beta}(u| \mathcal{S})$\;
%\ENDFOR
%\STATE $v\longleftarrow \underset{u \in V({\mathcal{G}}), \mathcal{C}(u) \leq \mathcal{B}}{argmax} \Delta_{\beta}(u| \mathcal{S})$\;
%\STATE $\mathcal{S} \longleftarrow \mathcal{S} \cup \{v\}$\;
\WHILE{$\exists u \in V(\mathcal{G}) \setminus \mathcal{S} \text{ and }\mathcal{C}(\mathcal{S} \cup \{u\}) \leq \mathcal{B}$}
\FOR {All $u \in V({\mathcal{G}}) \setminus \mathcal{S}$}
\STATE $Cur_{u}=False$
\ENDFOR
\WHILE{True}
\STATE $w \longleftarrow \underset{u \in V(\mathcal{G}), \mathcal{C}(\mathcal{S} \cup \{u\}) \leq \mathcal{B}}{argmax} \Delta_{\beta}(u| \mathcal{S})$
\IF{$Cur_{w}=True$}
\STATE $\mathcal{S}= \mathcal{S} \cup \{w\}$\;
\STATE $break$\;
\ELSE
\STATE $\Delta_{\beta}(w| \mathcal{S})=\beta(\mathcal{S} \cup \{w\})- \beta(\mathcal{S})$\;
\STATE $Cur_{w}=True$\;
\ENDIF
\ENDWHILE
\ENDWHILE
\STATE $return \ \mathcal{S}$
\end{algorithmic}
\end{algorithm}
 In Lemma \ref{Lemma:Submodularity}, it has been shown that the earned  benefit function $\beta(.)$ is sub\mbox{-}modular and this implies that the marginal gain in earned benefit for a non\mbox{-}seed node (say $u$) with respect to the seed set in $i$\mbox{-}th iteration ($\mathcal{S}^{i}$) will always be more than that of with respect to the seed set in $(i+1)$\mbox{-}th iteration. In Algorithm \ref{Algo:3}, in the first iteration of the while loop (Line $2$), the earned benefit by the nodes in $V(\mathcal{G})$ individually is computed, sorted them in descending order, and put the node with the highest individual earned benefit in the seed set. Now, in the next iteration on words, during the computation of the marginal gain of the non\mbox{-}seed nodes in descending order of their marginal earned benefit, as soon as we get a node, whose marginal gain in the current iteration is more than that in the previous iteration of the next node in the sorted list, then we include the first node and move to the next iteration. This is because, as the benefit function is sub\mbox{-}modular, even if we compute the marginal gain, earned benefit is computed for the second and the subsequent nodes, it cannot be more than the values in the previous iteration. This process is iterated, until the budget is exhausted. One important point to observe here is that, escaping the unnecessary benefit function evaluation does not result in loosing approximation guarantee in the quality of the selected seed set. This exploitation of the sub\mbox{-}modularity property results in significant improvement in the efficiency of our proposed methodology, as we observe in our experiments.
\subsection{Efficient Heuristic Solution}
Though Algorithm \ref{Algo:3} is quite efficient, it is not enough to deal with large real\mbox{-}life social networks. Here, we propose an efficient heuristic solution for the EBM Problem. Before stating the procedure, first we state one important aspect of social influence. In social networks, influence of a node is bounded within $2$ to $3$ hops, which is called as the \textit{influence zone} of a node \cite{tang2017influence} \cite{cha2009measurement}. According to Goel et al. \cite{goel2012structure}, in a diffusion cascade, less than $10 \%$ of the influenced nodes resides more than hop count $2$ from any seed node. These existing results reported in the literature motivate us to design algorithm considering the \textit{locality of influence} effect. Based on this principle, to influence a target node, there should be at least one seed node within a few hop count. In this context, we define \textit{h\mbox{-}hop neighbor} of a node as follows:
\begin{mydef}[$h$\mbox{-}hop Neighbor]For a node $u \in V(\mathcal{G})$, its $h$\mbox{-}hop neighbor is defined as the set of nodes that are at most at a distance  of $h$ from $u$ and denoted as $\mathcal{N}^{h}(u)$, i.e., $\mathcal{N}^{h}(u)=\{u_j \vert dist(u_ju) \leq h \}$.  
\end{mydef} 
Among the nodes present in $h$\mbox{-}hop neighbor set of a target node, there can be many nodes, whose influence probability to the target node is extremely low. Hence, those nodes probably cannot be able to influence the target node. To identify such nodes, it is important to compute the influence probability. For a target node $u_i$, here we describe the procedure for computing $\mathcal{P}_{u_j \rightarrow u_i}$, where $u_j \in \mathcal{N}^{h}(u_i)$. We construct the breadth first search tree upto depth $h$ rooted at the node $u_i$. Now, for any node $u_j$ other than root of the tree, the value of $\mathcal{P}_{u_j \rightarrow u_i}$ can be be computed by the following equation:
\begin{equation} \label{Eq:7}
\mathcal{P}_{u_j \rightarrow u_i}= [1- \underset{w \in \mathcal{N}(u_i)}{\prod}(1- \mathcal{P}_{u_j \rightarrow w})] \mathcal{P}_{w \rightarrow u_i}
\end{equation}
In Equation (\ref{Eq:7}), the value of $\mathcal{P}_{u_j \rightarrow w}$ can be recursively computed, until the child of the currently processing node is $u_j$. For details, please look into \cite{tang2017influence}. Now, it is easy to identify among the nodes in the $\mathcal{N}^{h}(u_i)$ which are effective for influencing the target user $u_i$. 
Here, we define the \textit{Effective $h$\mbox{-}hop neighbors} as follows.
\begin{mydef} [Effective $h$\mbox{-}hop Neighbors]
Given a target node $u_i \in \mathcal{D}$ and an $\alpha \in [0,1]$, the effective  $h$\mbox{-}hop neighbor(s) of $u_i$ is a subset of its $h$\mbox{-}hop neighbors and denoted as $\mathcal{N}^{h}_{E}(u_i)$. For $u_j \in \mathcal{N}^{h}(u_i)$, the node $u_j$ is an effective $h$\mbox{-}hop neighbor of the node $u_i$, if $\mathcal{P}_{u_j \rightarrow u_i} \geq \alpha$, i.e., $\mathcal{N}^{h}_{E}(u_i)=\{u_j \vert u_j \in \mathcal{N}^{h}(u_i) \wedge \mathcal{P}_{u_j \rightarrow u_i} \geq \alpha \}$.
\end{mydef} 
For any node say $u_j$, the main criterion to be included in the seed set is how much benefit it can earn. If the node is one of the target nodes, then the benefit associated with this node is surely be earned and along with this, if there are some target nodes (say $u_i$) within a few hop distance, benefit corresponding to that node may be earned, however, it depends upon the influence probability $\mathcal{P}_{u_j \rightarrow u_i}$. Now, we define the earned benefit of a node as follows.
\begin{mydef}[Earned Benefit of a Node]
For a node $u_j \in V(\mathcal{G})$, its earned benefit $\mathcal{EB}(u_j)$ is defined as the amount of benefit that can be earned by including this node in the seed set. It has two components. One is the direct benefit due to this node. The other one is the expected benefit due to influencing nearby target nodes. Mathematically, it can be expressed as follows:
\begin{equation} \label{Eq:8}
\mathcal{EB}(u_j)= b(u_j) + \underset{u_i \in \mathcal{N}^{h}(u_j) \wedge \mathcal{P}_{u_j \rightarrow u_i} \geq \alpha}{\sum} \mathcal{P}_{u_j \rightarrow u_i} . b(u_i)
\end{equation} 
\end{mydef}
There are two components in the right hand side of Equation (\ref{Eq:8}). The first part is due to the benefit associated with this particular node and the second part signifies the `expected earned benefit', i.e., the expected benefit due to the influence of the nodes within the few of distance of the node under consideration.  
\par Now, we describe the hop\mbox{-}based heuristic for solving the EBM Problem. First, we create an array for storing the expected earned benefit of each individual node and initialized with $0$ for non\mbox{-}target nodes and associated benefit value for the target nodes (from Lines $1$ to $6$ of Algorithm \ref{Algo:4}). Then, for a target node, we compute the  effective $h$\mbox{-}hop neighbors (from Lines $8$ to $14$). Then, for each of these nodes, we compute the expected benefit that can be earned by influencing the target node and sum it up. This process is repeated for each of the target nodes. Next, we divide the earned benefit of each target node by its selection cost and sort the nodes in descending order  based on this earned benefit value. Finally, we choose the seed node from this sorted list until the budget is exhausted. Algorithm \ref{Algo:4} describes this procedure.
\begin{algorithm}[h]
 \caption{A Hop\mbox{-}Based Heuristic for the EBM Problem}
 \label{Algo:4}
 \begin{algorithmic}[1]
 \renewcommand{\algorithmicrequire}{\textbf{Input:}}
 \renewcommand{\algorithmicensure}{\textbf{Output:} }
 \REQUIRE Social Network $\mathcal{G}(V, E, \mathcal{P})$, Target Nodes $\mathcal{D}$, Cost Function $\mathcal{C}$, Benefit Function $b$, Hop Count $h$, Cut off Probability $\alpha$, and Budget $\mathcal{B}$.  
 \ENSURE   The seed set $\mathcal{S} \subseteq V(\mathcal{G})$ such that $\underset{u \in \mathcal{S}}{\sum} \mathcal{C}(u) \leq \mathcal{B}$.  
 \STATE \text{Create Vector } ($\mathcal{EB}, n, 0)$ 
 \FOR{$All \ u \in V(\mathcal{G})$}
 \IF{$u \in \mathcal{D}$}
  \STATE $\mathcal{EB}(u)=\beta(u)$
  \ENDIF
 \ENDFOR
  \FOR{$Each \ u \in \mathcal{D}$}
  \STATE $\mathcal{N}^{h}(u)=\{v \vert dist(uv) \leq h\}$\;
  \FOR{Each $ w \in \mathcal{N}^{h}(u)$}
  \STATE $\text{Compute } \mathcal{P}_{w \rightarrow u}$ using Equation (\ref{Eq:7})\;
  \IF{$\mathcal{P}_{w \rightarrow u} \geq \alpha$}
\STATE $\mathcal{N}^{h}_{E}(u)=\mathcal{N}^{h}_{E}(u) \cup \{(w, \mathcal{P}_{w \rightarrow u})\}$\;
\ENDIF
\ENDFOR
\FOR{Each $ (w,\mathcal{P}_{w \rightarrow u}) \in \mathcal{N}_{E}^{h}(u)$}
\STATE $\mathcal{EB}(w)=\mathcal{EB}(w)+\beta(u).\mathcal{P}_{w \rightarrow u}$\;
  \ENDFOR
  \ENDFOR
\FOR{Each $ w \in V(\mathcal{G})$}
\STATE $\mathcal{EB}(w)=\frac{\mathcal{EB}(w)}{\mathcal{C}(w)}$\;
\ENDFOR
\STATE $V=\text{Sort }V(\mathcal{G}) \text{ based on } \mathcal{EB}$\;
\STATE $i=1$\;
\STATE $\mathcal{S}\leftarrow \phi$\;
\WHILE{$\mathcal{B}\geq 0$}
\IF{$\mathcal{B} \geq \mathcal{C}(V[i])$}
\STATE $\mathcal{S}=\mathcal{S} \cup \{V[i]\}$\;
\STATE $\mathcal{B}=\mathcal{B} - \mathcal{C}(V[i])$\;
\ENDIF
\STATE $i=i+1$\;
\ENDWHILE
\STATE $return \ \mathcal{S}$\; 
 \end{algorithmic} 
\end{algorithm}
\par Now, we analyze the time and space requirement of Algorithm \ref{Algo:4} by assuming it as a sparse and $d$\mbox{-}regular graph. For initializing the array $\mathcal{EB}$ requires $\mathcal{O}(n)$ time (Line $1$ to $6$). Now, for a target node $u_i$ in a $d$\mbox{-}regular graph, number of nodes and edges within the hop $h$ is $\mathcal{O}(d^{h+1})$. Hence, performing breadth first search from $u_i$ upto depth $h$ requires $\mathcal{O}(d^{h+1})$ time. For computing the influence probability  from each node $u_j \in \mathcal{N}^{h}(u_i)$ to $u_i$, i.e.,$\mathcal{P}_{u_j \rightarrow u_i}$ and comparing with $\alpha$ requires $\mathcal{O}(h d^{h})$ time. In the worst case, all the $h$\mbox{-}hop neighbor nodes may be the effective $h$\mbox{-}hop neighbor nodes. Then, for computing the earned benefit by influencing the target node $u_i$ requires $\mathcal{O}(d^{h+1})$ time. The same process is iterated over all the target nodes. Hence, the execution from Lines $7$ to $18$ of Algorithm \ref{Algo:4} requires $\mathcal{O}(\vert \mathcal{D} \vert d^{h+1}(h d^{h}+d^{h+1}))$. Dividing the earned benefit by the corresponding selection cost requires $\mathcal{O}(n)$ time (Line $19$ to $21$). Sorting the nodes based on this value requires $\mathcal{O}(n \log n)$ time. Now, scanning the sorted list for selecting the seed nodes requires $\mathcal{O}(n)$ time. Hence, total computational time of Algorithm \ref{Algo:4} is $\mathcal{O}(n + \vert \mathcal{D} \vert d^{h+1}(h d^{h}+d^{h+1})+ n + n \log n + n)$, which is equivalent to $\mathcal{O}(n \log n + \vert \mathcal{D} \vert d^{h+1}(h d^{h}+d^{h+1}))$. Other than the input social network, additional space requirements due to storing the earned benefits, influence probability and seed set which is of $\mathcal{O}(n)$, $\mathcal{O}(d^{h+1})$, and $\mathcal{O}(\vert \mathcal{S} \vert)$, respectively. Hence, the total space requirement of Algorithm \ref{Algo:4} is of $\mathcal{O}(n + d^{h+1})$. The formal statement is stated in Theorem \ref{Th:4}.
\begin{mythm}\label{Th:4} 
Algorithm \ref{Algo:4} has the running time of $\mathcal{O}(n \log n + \vert \mathcal{D} \vert d^{h+1}(h d^{h}+d^{h+1}))$ and space requirement of $\mathcal{O}(n + d^{h+1})$.
\end{mythm}
\section{Experimental Evaluation} \label{Sec:EE}
In this section, we report the experimental evaluation of our proposed methodologies. Initially, we start with a brief description of the datasets.
\subsection{Dataset Description} \label{Sec:DATASET}
In our experiments, we use the following four publicly available social network datasets. 
\begin{itemize}
\item \textbf{Email\mbox{-}Eu\mbox{-}core network Dataset} \footnote{\url{http://snap.stanford.edu/data/email-Eu-core.html}} \cite{yin2017local}, \cite{leskovec2007graph}: The network is generated based on the e\mbox{-}mail exchanges among different departments from a large European research institution. There is an edge between the users $u_i$ and $u_j$, if there is an e\mbox{-}mail exchange between them.
\item \textbf{Facebook Network Dataset} \footnote{\url{http://snap.stanford.edu/data/ego-Facebook.html}} \cite{leskovec2012learning}: This dataset was collected from survey participants using a Facebook app. Each user of the network is represented by a node, and two vertices are connected by an edge, if the corresponding users are friend of each other in Facebook.
\item \textbf{Physics Network Dataset} \footnote{\url{https://arxiv.org/}}: This is an academic collaboration network among the researchers of physics section of \url{arxiv.org}. Two users are connected by an undirected edge, if they are co\mbox{-}author in atleast one paper.
\item \textbf{Epinions} \footnote{\url{http://www.epinions.com/?sb=1}} \cite{richardson2003trust}: This is a who\mbox{-}trust\mbox{-}whom on\mbox{-}line social network of a general consumer review site: \url{Epinions.com}. There is a directed edge from the user $u_i$ to $u_j$, if the user $u_i$ trusts $u_j$.
\end{itemize}
Among them, the first, second and fourth one are downloaded from \textit{Stanford Social Network Analysis} \url{http://snap.stanford.edu/data/index.html} and the third one is from \url{https://www.microsoft.com/en-us/research/people/weic/#!selected-projects}. Here, we give a brief description of each of the datasets.These datasets have been extensively used in social influence maximization research \cite{chen2010scalable}. Table \ref{Tab:2} gives a basic statistics of the described datasets.
%(such as number of nodes and edges, average degree and average clustering co\mbox{-}efficient) 
\begin{center}
\begin{table}[H]
\caption{Basic statistics of the datasets.}

\centering
\begin{tabular}{|l|l|l|l|l|}
 \hline
\textit{Dataset Name} & $\vert V(G) \vert$ & $\vert E(G) \vert$ & \textit{Avg Deg} & \textit{Avg Clus Coeff} \\
\hline
Email-Eu-core network & 1005 & 25571 & 25.443 & 0.3994 \\
\hline
Facebook Dataset & 4039 & 88234 & 43.6910 & 0.6055 \\
\hline
PHY Network & 37154 & 231584 & 12.466 & 0.2371\\
\hline
Epinions & 75879 & 508837 & 15.6345 & 0.1378 \\
\hline
\end{tabular}
\label{Tab:2}
\end{table}
\end{center}  
\subsection{Parameter Settings}
\subsubsection{Diffusion Probability} In this paper, we consider the following two  diffusion probability settings. 
\begin{itemize}
\item \textit{Uniform Setting}: In this setting, $\forall (u_iu_j) \in E(\mathcal{G})$, $\mathcal{P}_{u_i \rightarrow u_j}=p_{c}$ and $p_c \in (0,1]$. We set the value of $p_c$ as $0.1$.  This value has also been used in the literature, in many studies \cite{goyal2011celf++}.
\item \textit{Trivalency Setting}: In this setting, each edge is assigned diffusion probability uniformly at random from the set $\{0.1, 0.01, 0.001\}$. 
\end{itemize}

\par On the other hand, 
\subsubsection{Target Nodes}
In this study, we select $20\%$ of the nodes as target nodes and they are chosen uniformly at random. We adopt this settings from \cite{nguyen2017billion}.  
\subsubsection{Cost and Benefit} 
In this study, we follow two different settings.
\begin{itemize}
\item First one is the random setting, where the selection cost of the nodes and the earned benefit of the target nodes are selected from the intervals $[1,50]$ and $[50,100]$, respectively, uniformly at random. We adopt this setting from \cite{nguyen2013budgeted} and call it the random setting.
\item Secondly, the influence ability of a node is directly proportional to its degree. Naturally, the selection cost of a node should be proportional to its degree. We adopt another settings from \cite{nguyen2017billion}. By this setting, we compute the selection cost of the node $u_i$ is computed as 
\begin{equation}
\mathcal{C}(u_i)= \frac{n \ deg(u_i)}{2.m},
\end{equation}
and in this case, the benefit of each target node is considered as $1$. We call this setting as the `degree proportional' setting.
\end{itemize}  
\subsubsection{Budget} In our study, based on the two different cost assignment settings, we adopt two different budget settings. In case of random setting, we consider the budget values starting from $2000$ continued till $16,000$, and each time is incremented by $2,000$. In the `degree proportional setting', we start with the budget value of $100$ and continued until $800$ with a gap of $100$.
\subsubsection{Hop Count and Cut Off Probability}
In Algorithm \ref{Algo:4}, we use a hop count $h$ and cut off probability $\alpha$ for computing the effective nodes. In our experiments, we choose the value of $h$ as $2$ and the value of $\alpha$ as $0.1$. We adopt these settings from \cite{tang2017influence}. 
\subsection{Algorithms in the Experiment}
Here, we have listed out the algorithms that we have listed out for the  experimentation.
\subsubsection{Algorithms proposed in this paper}
\begin{itemize}
\item \textbf{Incremental Greedy Approach with Approximation Guarantee} (IGAAG): This is basically the Algorithm \ref{Algo:1a} of this paper, which returns either the set of nodes chosen incrementally by Algorithm \ref{Algo:1} or the node that causes maximum individual benefit gain.
\item \textbf{Incremental Greedy Approach with Improved Performance} (IGAIP):
This is the Algorithm \ref{Algo:3} of this paper, which improves the Algorithm \ref{Algo:1a} by exploiting the sub\mbox{-}modularity property of the benefit function.
\item \textbf{Hop\mbox{-}Based Heuristic} (HBH): This is the Algorithm \ref{Algo:4} of this paper, which works based on the computation of expected earned benefit of the nodes that are within the $h$\mbox{-}hop (for a given value of $h$) of the target nodes.
\end{itemize}
\subsubsection{Baseline Algorithms}
\begin{itemize}
\item \textbf{Maximum Degree Heuristic (Max\_DEG)}: In this method, the maximum degree nodes within the budget is returned as the seed set. This method has been used in previous studies as well \cite{kempe2003maximizing}.
\item \textbf{Degree Discount Heuristic (DEG\_DIS)}: This is a popular heuristic for the SIM Problem proposed by Chen et al. \cite{chen2009efficient}. In this heuristic, if $u$ is a seed node and $(uv) \in E(\mathcal{G})$, then the degree of $v$ will be discounted by $2t_{v}+(d_{v}-t_{v})t_{v}\mathcal{P}_{u \rightarrow v}$, where $t_{v}$ is the number of neighbors of $v$ currently in the seed set, and $d_{v}$ is the degree of $v$. This method has been used in many previous studies \cite{jiang2011simulated}.
\item \textbf{Single Discount Heuristic (SIN\_DEG)}: This a variant of degree discount heuristic proposed by Chen et al. \cite{chen2009efficient}. In this heuristic, if $u$ is a seed node and $(uv) \in E(\mathcal{G})$, then the degree of $v$ will be discounted by $1$. This method has been used as a baseline in previous studies \cite{cao2011maximizing} \cite{bucur2016influence}.
\item \textbf{Prefix excluded Maximum Influence Arbarence (PMIA)}: This is one of the state\mbox{-}of\mbox{-}the\mbox{-}art and popular heuristic for influence maximization problem proposed by Chen et al. \cite{chen2010scalable} \cite{wang2012scalable}. 
\item \textbf{ComPBRA}: This is a recently developed community\mbox{-}based solution framework for the EBM Problem developed by Banerjee et al. \cite{banerjee2019maximizin}.
\end{itemize}
All the algorithms have been implemented in `Python 3.4' along with `NetworkX 1.9.1'. We have carried out all the experiments in a high performance computing cluster with $5$ nodes and each of them having $64$ cores and $64 \ GB$ of RAM running in Centos $6.7$ environment. As, the Algorithm \ref{Algo:1a} (IGAAG) is quite inefficient, we don't execute this on the larger datasets (e.g., Physics Network Dataset, Epinions).
\subsection{Experimental Results and Discussion}
The main goal of our experimentation is to make a comparative study of the proposed as well as baseline methods in terms of performance. It is measured as the amount of earned benefit obtained by influencing the target users due to the initial activation of the seed nodes selected by different algorithms. We also report the computational time requirement by different algorithms for selecting the seed sets. 
\subsubsection{Performance on Earned Benefit}

First row of Figure \ref{Fig:Results} shows the budget vs. earned benefit plot for the `email-Eu-core' dataset. Based on the random and degree proportional setting, the maximum benefit that can be earned is $13912$ and $179$, respectively. From the results, it is observed that there is a gap in the earned benefit between the existing methods and the methods proposed in this paper. The gap is even significant in tri\mbox{-}valency setting compared to the uniform setting. As an example, in uniform setting ($p_c=0.1$) with random cost and benefit assignment for $\mathcal{B}=16000$, among the existing methods from the literature, the seed set selected by ComPBRA leads to more earned benefit and the amount is $12231$, which is $73.34 \%$ of the maximum possible. On the other hand, among the proposed methodologies, the seed set selected by the IGAAG leads to more amount of earned benefit $13912$, which is $91.37\%$ of maximum possible. In degree proportional setting, for $\mathcal{B}=16000$, in tri\mbox{-}valency setting among the existing methods, the seed set selected by both PMIA and ComPBRA leads to the earned benefit of $167$ ($83.5 \%$ of the maximum possible), whereas the same for both IGAAG and IGAIP is $172$ ($86 \%$ of the maximum possible).

Next, we report the results for the `Facebook' dataset in the second row of Figure \ref{Fig:Results}. In this dataset also, we observe that the seed set selected by the proposed methodologies leads to more earned benefit compared to the existing methods. As an example, when the budget value is $16000$, under random cost and benefit with tri\mbox{-}valency setting, among the existing methods the seed set selected by ComPBRA leads to the earned benefit of $15578$. However, the same due to the seed set selected by the proposed hop\mbox{-}based heuristic is $20450$, which is almost $31 \%$ more. Now, under the degree proportional cost and tri\mbox{-}valency setting, when the budget value is $16000$, among the existing methods the earned benefit due to the seed set by ComPBRA is $328$, and the same by the hop\mbox{-}based heuristics is $426$.
%other than the random cost and benefit assignment under uniform setting, for most of the budget values, the seed set selected by the proposed methodologies leads to more amount of earned benefit compared to the existing methods. However, in random cost and benefit assignment, the earned benefits due to the seed sets selected by IGAAG and IGAIP are more compared to that of the MAX\_DEG, DEG\_DIS and SIN\_DIS methods and are less compared to that of MIA and PMIA methods. However, the hop\mbox{-}based heuristic is the best in this case.
\begin{figure*}
\centering
\begin{tabular}{cccc}
\includegraphics[height=4 cm, width=4.5 cm]{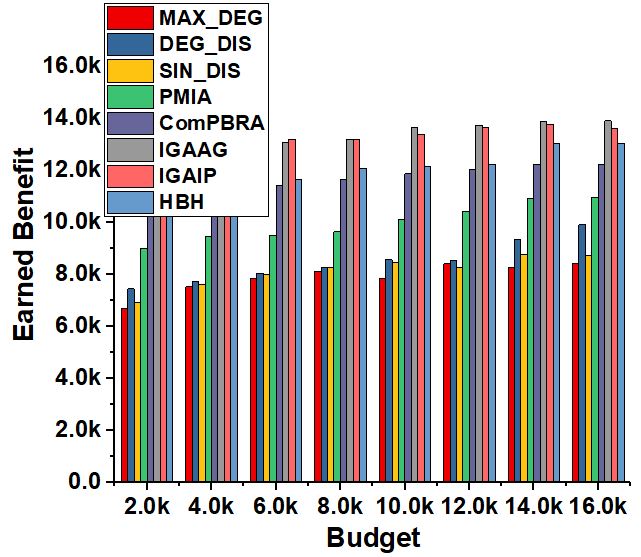} & \includegraphics[height=4 cm, width=4.5 cm]{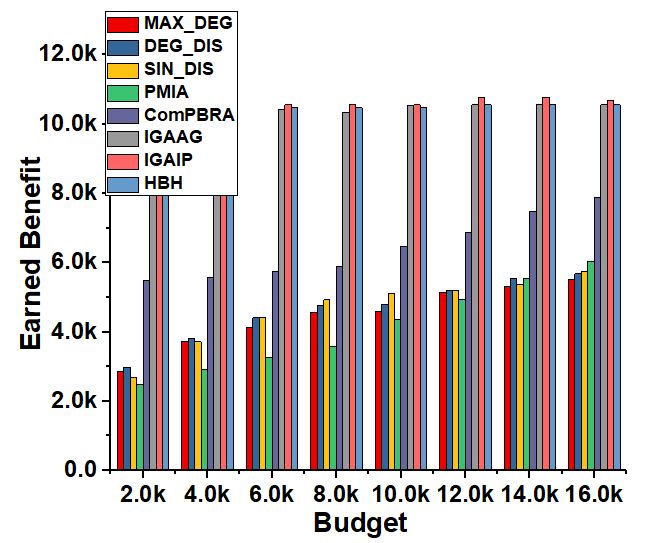} & \includegraphics[height=4 cm, width=4.5 cm]{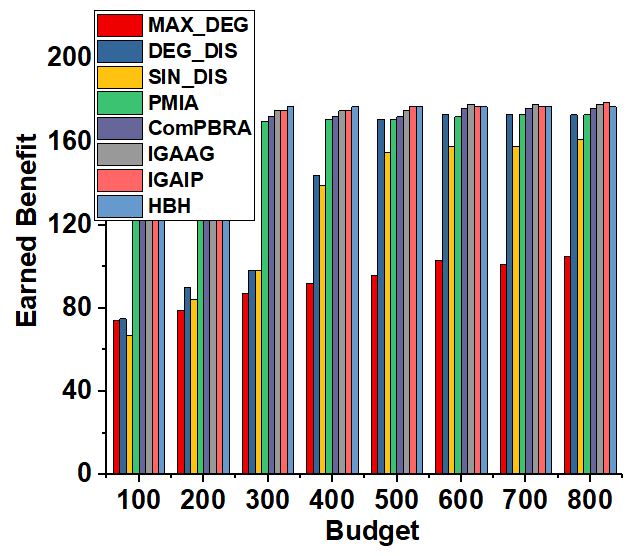} & \includegraphics[height=4 cm, width=4.5 cm]{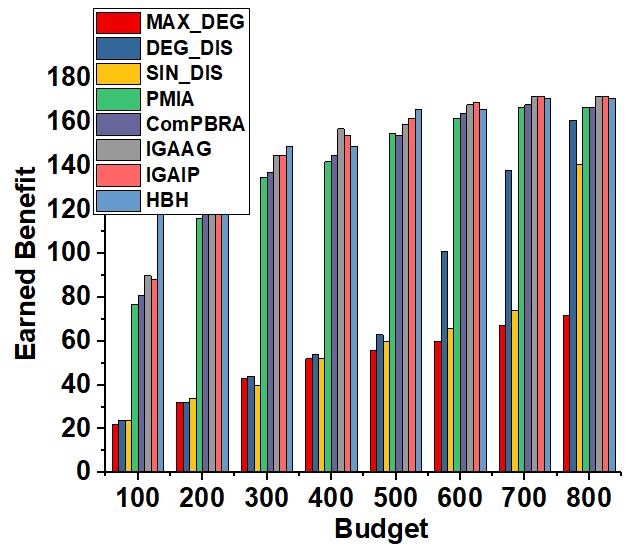} \\
(1UR)    & (1TR)  & (1UD)   & (1TD)  \\

\includegraphics[height=4 cm, width=4.5 cm]{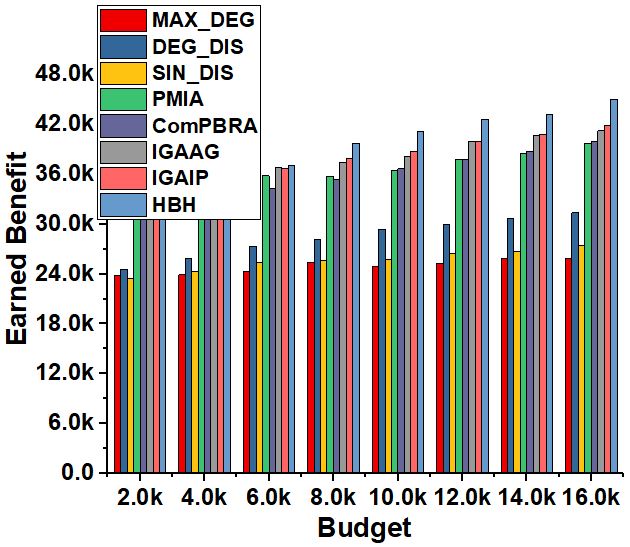} & \includegraphics[height=4 cm, width=4.5 cm]{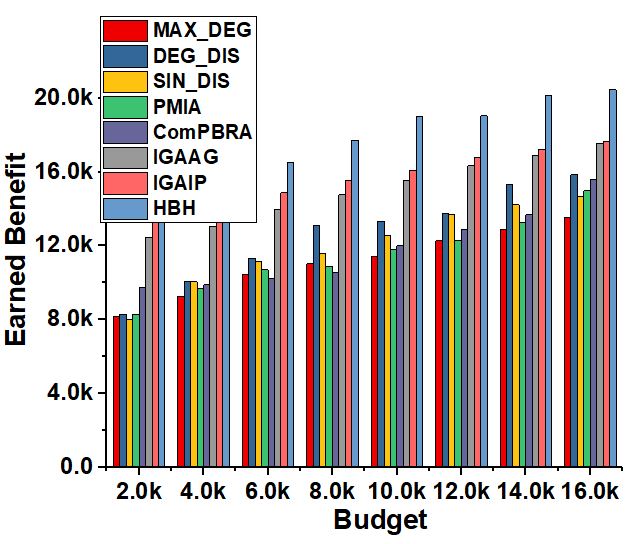} & \includegraphics[height=4 cm, width=4.5 cm]{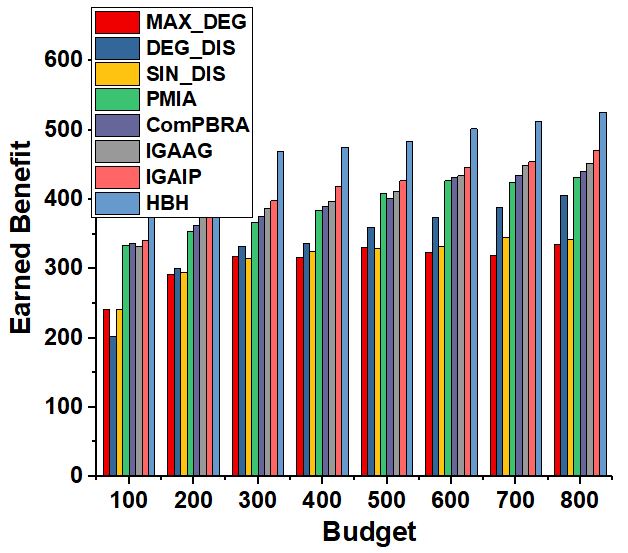} & \includegraphics[height=4 cm, width=4.5 cm]{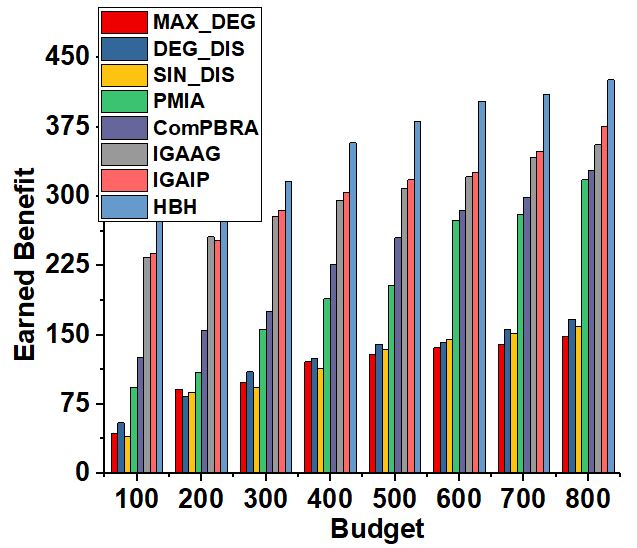} \\
(2UR)    & (2TR)  & (2UD)  & (2TD) \\

\includegraphics[height=4 cm, width=4.5 cm]{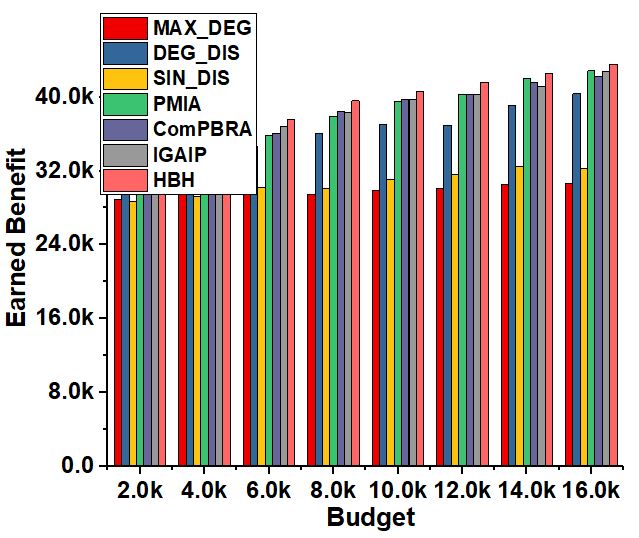} & \includegraphics[height=4 cm, width=4.5 cm]{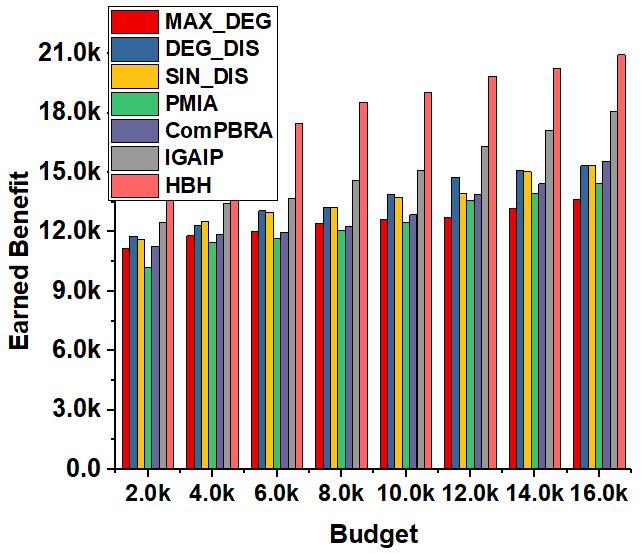} & \includegraphics[height=4 cm, width=4.5 cm]{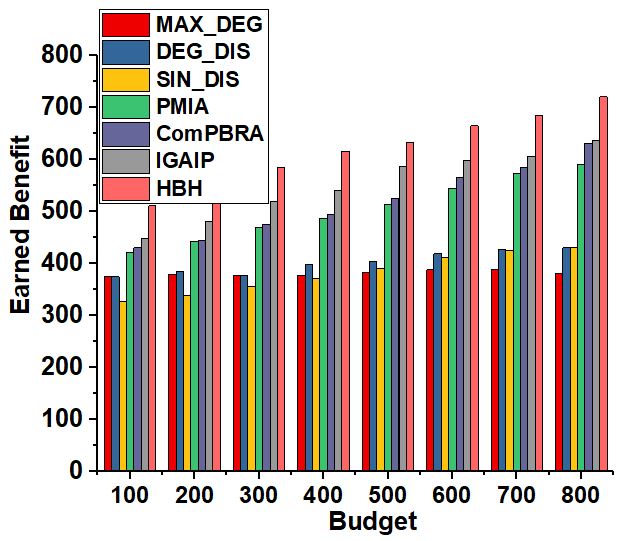} & \includegraphics[height=4 cm, width=4.5 cm]{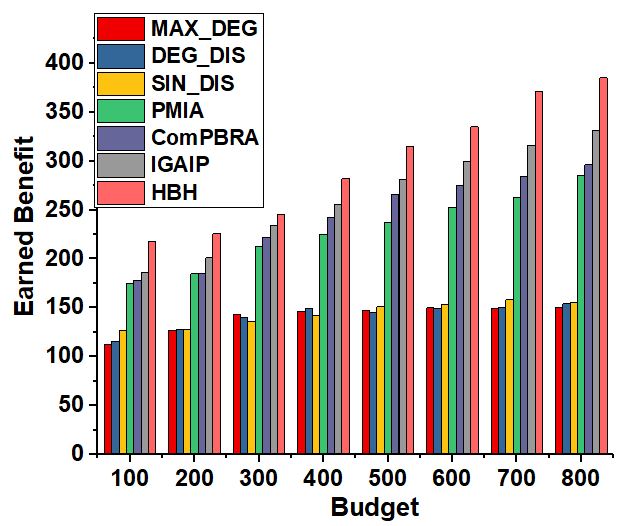} \\
(3UR) & (3TR) & (3UD)  & (3TD) \\

\includegraphics[height=4 cm, width=4.5 cm]{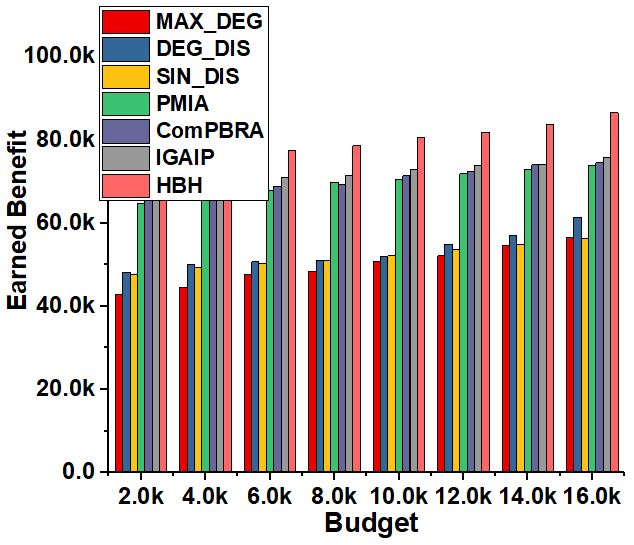} & \includegraphics[height=4 cm, width=4.5 cm]{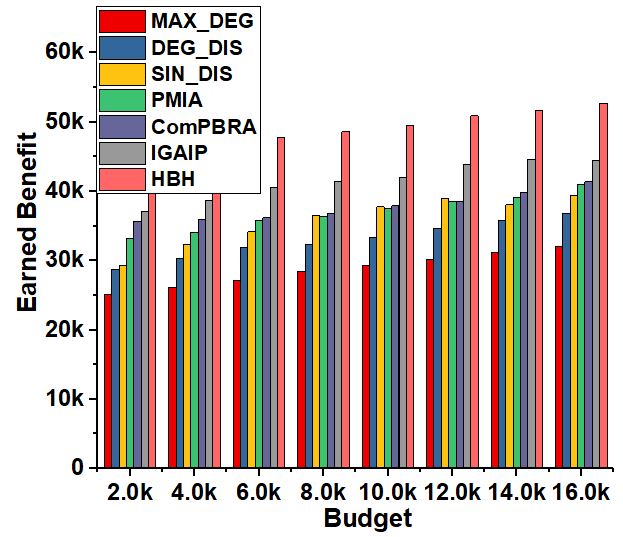} & \includegraphics[height=4 cm, width=4.5 cm]{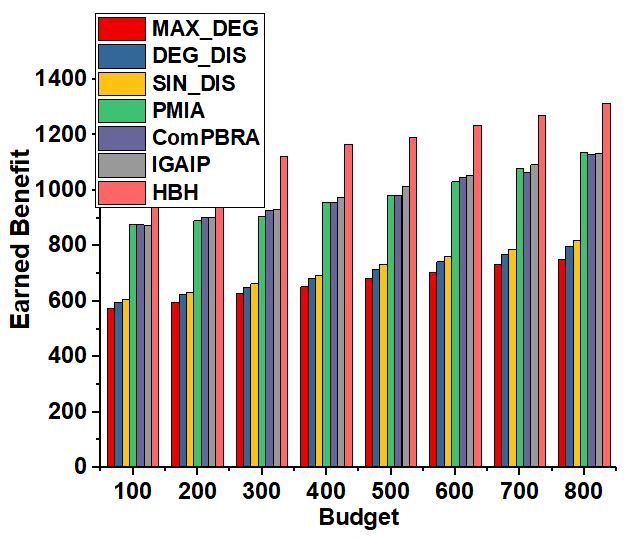} & \includegraphics[height=4 cm, width=4.5 cm]{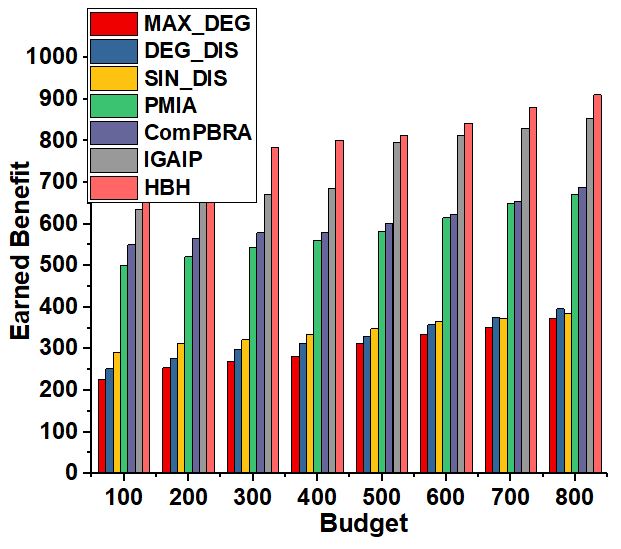} \\
(4UR)  & (4TR)  & (4UD)  & (4TD)  \\
\end{tabular}
\caption{Budget vs. Earned Benefit Plots for different datasets. In the individual figure captions 1, 2, 3, 4 denotes the datasets in which they have been described in Section \ref{Sec:DATASET}. U and T denotes the uniform and trivalency probability setting. R and D denotes random and degree proportional cost setting.}
\label{Fig:Results}
\end{figure*}
Next, we show the results for the `Physics Collaboration Network' dataset in the third row of Figure \ref{Fig:Results}. In this dataset also, we observe a significant gap in the earned benefit between the existing methods and the methods proposed in the literature. The gap is more in case of tri\mbox{-}valency setting. As an example, for $\mathcal{B}=16000$, under random cost with uniform influence probability setting, among the existing methods, the seed set selected by PMIA leads to maximum amount of earned benefit which is $42817$ and the same by the hop\mbox{-}based heuristic is $43568$. In tri\mbox{-}valency setting, for $\mathcal{B}=16000$, the seed set selected by SIN\_DIS leads to the earned benefit, which is equal to $15352$, and the same obtained by hop\mbox{-}based heuristic is $20955$. This is approximately $36\%$ more compared to the SIN\_DIS method. 

Next, we report the results for the `Epinions' dataset in the last row of Figure \ref{Fig:Results}. In this dataset also, we observe a significant difference between the earned benefit due to the seed sets selected by the baseline methods and the methods proposed in this paper. As an example, for the highest budget ($\mathcal{B}=16000$), in uniform setting under random cost and benefit assignment seed set selected by the ComPBRA leads to the earned benefit of worth $74561$, and the same in case of the `hop\mbox{-}based heuristic' is $86564$, which is almost $16 \%$ more compared to the ComPBRA. Similarly, in degree proportional setting under tri\mbox{-}valency diffusion model, the seed set selected by the ComPBRA leads to the earned benefit of amount $688$ and the same for the `hop\mbox{-}based heuristic' is $911$, which is almost $32\%$ more compared to the ComPBRA.

From the results, it is observed that the seed set selected by the proposed methodologies can lead to more amount of earned benefit compared to the existing methods considered in this paper. Next, we report the computational time of the proposed and baseline methods.
\subsubsection{Computational Time}
Table \ref{Tab:3} reports the computational time required for selecting the seed sets by different methodologies. From the reported results, it is observed that though the IGAAG method can achieve an approximation guarantee,  the computational time requirement is quite impractical. However, the IGAIP method overcomes this issue by improving it upto $220$ times faster compared to IGAAG. However, the hop\mbox{-}based heuristic is much more efficient and also scalable compared to both IGAAG and IGAIP, while achieving the almost similar amount of earned benefit, in some instances even more. 
\par Among the baseline methods, the MAX\_DEG is the fastest one, as it returns the high degree nodes within the budget. The DEG\_DIS and SIN\_DIS methods take more time compared to the MAX\_DEG method. Among the existing methods, PMIA is seen to be the fastest.
\par Now, in real\mbox{-}life applications of this problem, such as `computational advertisement', `viral marketing' etc. from the advertisers point of view, the main priority will be the earned benefit. However, the methodology used for seed set selection purpose should be able to perform this task with a reasonable computational time. From the experimental evaluation, it is established that among the proposed methodologies, the hop\mbox{-}based heuristic is far ahead compared to the existing methods.     
\begin{table}
\caption{Computational Time Requirement for the Proposed as well as Baseline Methods}
\label{Tab:3}
\resizebox{.915\hsize}{!}{
		\begin{tabular}{|c| c|| c|c|c|c|c|c|c | c|}
			\hline
			\multirow{2}{*}{Dataset} & \multirow{2}{*}{Budget} & \multicolumn{8}{c||}{Computational Time (in Seconds)} \\
			\cline{3-10}
			 &  & IGAAG & IGAIP & HBH & MAX\_DEG & DEG\_DIS & SIN\_DIS & PMIA & ComPBRA \\
			\hline
			\hline
			\multirow{8}{*}{\textbf{Email}}&2000  & 6.2351 $\times 10^{3}$ & 35.5362 & 0.0614 & 0.0253 & 0.0293 & 0.2825 & 0.2671 & 0.1667 \\
			
			& 4000 & 6.4995 $\times 10^{3}$ & 61.2136 & 0.2222 & 0.0269 & 0.0358 & 0.4567 & 0.4988 & 0.3911 \\
			
			& 6000 & 6.4995 $\times 10^{3}$ & 60.8463 & 0.2327 & 0.0294 & 0.0459 & 0.2289 & 1.0202 & 0.6006 \\
			
			& 8000 & 6.6856 $\times 10^{3}$ & 63.1582 & 0.2826 & 0.0328 & 0.0830 & 0.5065 & 1.1021 & 0.7407 \\
			
			& 10000 & 6.8265 $\times 10^{3}$ & 66.7364 & 0.4451 & 0.0365 & 0.1216 & 0.4219 & 1.2865 & 0.9168 \\
			
			& 12000 & 7.0004 $\times 10^{3}$ & 79.1924 & 0.7923 & 0.0416 & 0.1461 & 0.6986 & 1.3761 & 1.0451 \\
			
			& 14000 & 7.3358 $\times 10^{3}$ & 94.2375 & 1.2280 & 0.0474 & 0.1669 & 0.5504 & 1.4902 & 1.1184 \\
			
			& 16000 & 7.5138 $\times 10^{3}$ & 110.7938 & 1.5740 & 0.0548 & 0.1872 & 0.5643 & 1.9567 & 1.2251 \\
			\hline
			\multirow{8}{*}{\textbf{Facebook}}& 2000 & 9.3518 $\times 10^{3}$ & 57.5381 & 0.5593 & 0.1270 & 0.1371 & 0.1779 & 0.5124 & 0.3252 \\
			
			& 4000 & 1.0031 $\times 10^{4}$ & 72.1473 & 2.7052 & 0.1291 & 0.1411 & 0.1801 & 0.6301 & 0.4236 \\
			
			& 6000 & 1.2436 $\times 10^{4}$ & 88.9735 & 2.8905 & 0.1351 & 0.1544 & 0.1521 & 0.9325 & 0.7095 \\
			
			& 8000 & 1.4835 $\times 10^{4}$ & 96.5408 & 3.5700 & 0.1426 & 0.1709 & 0.1596 & 1.4002 & 1.1797 \\
			
			& 10000 & 1.6124 $\times 10^{4}$ & 114.8327 & 4.5553 & 0.1495 & 0.1901 & 0.1937 & 1.7522 & 1.5479 \\
			
			& 12000 & 1.7831 $\times 10^{4}$ & 116.1186 & 14.8733 & 0.1566 & 0.2128 & 0.1770 & 3.8360 & 3.5625 \\
			
			& 14000 & 1.9149 $\times 10^{4}$ & 142.9568 & 19.1873 & 0.1654 & 0.2391 & 0.2164 & 5.1924 & 4.8920 \\
			
			& 16000 & 2.1285 $\times 10^{4}$ & 145.1749 & 8.2808 & 0.1742 & 0.2778 & 0.1981 & 13.1280 & 7.5171 \\
			\hline
						\multirow{8}{*}{\textbf{Physics}}& 2000 & - & 336.7986 & 1.5963 & 1.5816 & 1.6041 & 2.8480 & 7.2921 & 4.8043  \\
			
			& 4000 & - & 398.8845 & 1.9457 & 1.3813 & 1.6682 & 1.8261 & 6.4076 & 7.1599 \\
			
			& 6000 & - & 424.1447 & 2.6486 & 1.8345 & 1.4318 & 2.4793 & 8.0311 & 9.4225 \\
			
			& 8000 & - & 464.6438 & 3.3497 & 1.8320 & 1.9460 & 1.9653 & 11.5990 & 11.9076 \\
			
			& 10000 & - & 488.6438 & 5.5643 & 1.4763 & 2.0533 & 2.2705 & 16.5869 & 12.9141 \\
			
			& 12000 & - & 531.4116 & 10.3245 & 1.5406 & 0.8396 & 2.2997 & 16.7783 & 16.0735 \\
			
			& 14000 & - & 558.6329 & 15.6457 & 2.1607 & 1.9682 & 3.1089 & 18.4652 & 18.6298 \\
			
			& 16000 & - & 602.1542 & 17.8947 & 2.4451 & 2.3792 & 1.5399 & 23.2580 & 17.5282 \\
			\hline
						\multirow{8}{*}{\textbf{Epinions}}& 2000 & - & 751.3667 & 60.5972 & 8.7695 & 10.2764 & 11.8725 & 46.7235 & 44.2557 \\
			
			& 4000 & - & 789.6519 & 66.8945 & 8.7535 & 10.9163 & 11.5137 & 59.1578 & 58.7692 \\
			
			& 6000 & - & 797.6386 & 68.8924 & 9.0238 & 11.2865 & 11.4792 & 64.6349 & 57.1139 \\
			
			& 8000 & - & 812.9137 & 73.2648 & 9.1369 & 11.1869 & 12.3527 & 81.9739 &  59.1975\\
			
			& 10000 & - & 842.7459 & 76.5489 & 9.2759 & 12.0237 & 12.8573 & 86.7682 & 67.9834 \\
			
			& 12000 & - & 852.8564 & 82.4392 & 9.1349 & 12.9768 & 12.5737 & 81.2854 & 65.1158 \\
			
			& 14000 & - & 865.5867 & 78.3267 & 9.4672 & 12.5549 & 13.9136 & 80.4375 & 79.8859 \\
			
			& 16000 & - & 904.8127 & 86.5197 & 9.7959 & 13.7339 & 13.54879 & 82.1472 & 78.9657 \\

			\hline

		\end{tabular}
}
	\end{table}

\section{Conclusion and Future Direction}\label{Sec:CFD}
In this paper, we have studied the `Earned Benefit maximization problem', where a subset of nodes of the input social network are designated as target nodes and each of them is associated with a benefit value. Each node of the network is associated with a selection cost and the seed selection has to be done within an allocated budget with an aim to maximize the earned benefit. For this problem, we propose an $(1-\frac{1}{\sqrt{e}})$ factor approximation algorithm. By exploiting the sub\mbox{-}modularity of the benefit function, we improve the efficiency of this algorithm. To deal with the large scale social networks, we propose a hop\mbox{-}based heuristic solution for this problem. Reported results demonstrate that the seed set selected by the proposed methodologies leads to more amount of earned benefit compared to the existing methods. Now, this study can be extended in several directions. First of all, our study can be carry forwarded by considering the time varying nature of the real\mbox{-}world social networks. Secondly, the approximation bound that we have provided for our proposed methodology can be improved by more sophisticated analysis. Moreover, it will be interesting to come up with a game theoretic model of this problem.
\ifCLASSOPTIONcaptionsoff
  \newpage
\fi



\bibliographystyle{IEEEtran}
\bibliography{mybibfile}
\end{document}